\documentclass[reqno, eucal]{amsart}

\usepackage[mathscr]{eucal} 
\usepackage[dvips]{graphicx}
\usepackage[dvips]{color}
\usepackage{amsmath,amsfonts,amssymb,amsthm,amscd}
\usepackage{epic}
\usepackage{eepic}
\usepackage{longtable}
\usepackage{array}
\usepackage{here}
\numberwithin{equation}{section}

\newtheorem{theorem}{Theorem}[section]

\newtheorem{proposition}[theorem]{Proposition}
\newtheorem{corollary}[theorem]{Corollary}

\theoremstyle{definition}

\newtheorem{example}[theorem]{Example}
\newtheorem{remark}[theorem]{Remark}

\newcommand{\Z}{{\mathbb Z}}

\newcommand{\C}{{\mathbb C}}
\newcommand{\st}{\!\otimes\!}

\newcommand{\grbold}[1]{\mbox{\boldmath{$#1$}}}
\renewcommand{\L}{{\mathscr L}}
\newcommand{\M}{{\mathscr M}}
\newcommand{\ot}{\otimes}

\begin{document}

\title[Multispecies TASEP]{Multispecies TASEP and the tetrahedron equation}

\author{Atsuo Kuniba}
\email{atsuo@gokutan.c.u-tokyo.ac.jp}
\address{Institute of Physics, University of Tokyo, Komaba, Tokyo 153-8902, Japan}

\author{Shouya Maruyama}
\email{maruyama@gokutan.c.u-tokyo.ac.jp}
\address{Institute of Physics, University of Tokyo, Komaba, Tokyo 153-8902, Japan}

\author{Masato Okado}
\email{okado@sci.osaka-cu.ac.jp}
\address{Department of Mathematics, Osaka City University, 
3-3-138, Sugimoto, Sumiyoshi-ku, Osaka, 558-8585, Japan}


\maketitle

\vspace{0.5cm}
\begin{center}{\bf Abstract}
\end{center}
We introduce a family of layer to layer transfer matrices
in a three-dimensional (3D)  lattice model which can be viewed as 
partition functions of the $q$-oscillator valued six-vertex model
on $m \times n$ square lattice.
By invoking the tetrahedron equation we establish
their commutativity and bilinear relations mixing various boundary conditions.
At $q=0$ and $m=n$, they ultimately yield a new proof of 
the steady state formula for 
the $n$-species totally asymmetric simple exclusion process (TASEP)
obtained recently by the authors, revealing the 3D integrability in 
the matrix product construction.
\vspace{0.4cm}

\section{Introduction}\label{sec:intro}
Totally asymmetric simple exclusion process (TASEP) 
is a model of non-equilibrium stochastic dynamics
in physical, biological and many other systems.
It has been studied extensively in the last few decades
especially in one-dimension, 
which has led to numerous generalizations and analytical results. 
See for example \cite{BE, BP} and references therein.

By $n$-species TASEP or $n$-TASEP for short 
we mean in this paper 
the TASEP on a one-dimensional periodic chain $\Z_L$ with $L$-sites 
in which local states $\sigma_i$ take values in 
$\{0,1,\ldots, n\}$ and neighboring pairs 
$(\sigma_i, \sigma_{i+1})= (\alpha,\beta)$ with 
$\alpha>\beta$ are interchanged to
$(\beta,\alpha)$ with a uniform transition rate.

The main theme of the present paper, which is a 
continuation of \cite{KMO}, 
is the 3D integrability of the $n$-TASEP 
connected to the {\em tetrahedron equation} \cite{Zam80},
a 3D generalization of the Yang-Baxter equation \cite{Bax}.
It becomes visible and natural for the multispecies case $n\ge 2$.

In \cite{KMO}, combinatorial construction 
of the steady state probability 
${\mathbb P}(\sigma_1,\ldots, \sigma_L)$ of 
the $n$-TASEP by Ferrari-Martin \cite{FM} was identified 
with a composition of the {\em combinatorial $R$} \cite{NY}.
It is a quantum $R$ matrix of $U_q(\widehat{sl}_L)$ at $q=0$ where 
the original periodic chain $\Z_L$ has been incorporated into the Dynkin diagram
of the relevant quantum group.
It has led to a new matrix product formula 
\begin{align}\label{mnm}
{\mathbb P}(\sigma_1,\ldots, \sigma_L) 
= \mathrm{Tr}(X_{\sigma_1}\cdots X_{\sigma_L})
\end{align}
by applying the recent matrix product construction of the $R$ matrix 
based on the tetrahedron equation \cite{KOS}.
The result in \cite{KMO} possesses distinct features from the other ones 
\cite{EFM,PEM,CDW}.
The operator $X_\sigma$ itself is expressed as a configuration sum 
for a {\em corner transfer matrix} \cite{Bax} 
of the $q=0$-oscillator valued five-vertex model.
See (\ref{ngm}).
It serves as a layer to layer transfer matrix to constitute 
$\mathrm{Tr}(X_{\sigma_1}\cdots X_{\sigma_L})$
as a {\em partition function} of a 3D lattice model. 

Our aim in this paper is to elucidate a further 3D integrability
concerning an alternative approach to establish (\ref{mnm}).
It is by the so called  {\em hat relation}
\begin{align}\label{ers}
\sum_{0 \le \gamma,\delta \le n}
h^{\alpha,\beta}_{\gamma,\delta}X_\gamma X_\delta
= X_\alpha {\hat X}_\beta - {\hat X}_\alpha X_\beta\qquad
(0 \le \alpha, \beta \le n),
\end{align}
where $h^{\alpha,\beta}_{\gamma,\delta}$ is an element of the 
local Markov matrix defined in (\ref{H}) and (\ref{hdef}).
Construction of such 
companion operators ${\hat X}_0, \ldots, {\hat X}_n$ is a sufficient 
task to prove (\ref{mnm})\footnote{
as long as the right hand side is convergent} as is well known \cite{BE}.
See also section \ref{subsec:mp}.
We construct ${\hat X}_i$ similarly to $X_i$ 
as a weighted configuration sum as in (\ref{ngm})
and present a self-contained proof of the hat relation (\ref{ers}).
Our strategy is to upgrade the statement ultimately 
by introducing $q$-deformation, spectral parameters and embedding into a 
3D lattice model until the point where all the nonlocal commutation relations 
can be understood most naturally 
as a consequence of the single and local tetrahedron equation.
The analysis fully demonstrates 
the 3D integrable aspect of the steady state in the $n$-TASEP 
as promised in \cite{KMO}. 

The paper is organized as follows.
In section \ref{sec:tasep} we recall the $n$-TASEP and 
the steady state result in \cite{KMO}.
The operators $X_i$ and ${\hat X}_i$ are defined and 
the main statement, the hat relation, 
is formulated (theorem \ref{th:mho}).
In section \ref{sec:LM} we introduce the deformation parameter $q$ 
and define the 3D $L$ and $M$ operators.
Eigenvectors of the latter and the tetrahedron equation
among $L$ and $M$ (theorem \ref{prop:tetrahedron}) are described.
These contents serve as the local information controlling 
more involved nonlocal objects 
considered in the subsequent sections.
In section \ref{sec:LLT} we consider the 3D lattice model
associated with the 3D $L$ operator.
A family of layer to layer transfer matrices labeled with 
mixed boundary conditions $S(z)^{\bf a}_{\bf j}$ are introduced.
It is shown that each of them form a commuting family 
by invoking the tetrahedron equation 
and the eigenvectors of $M$ (proposition \ref{prop:S comm}).
In section \ref{sec:fb} we extend the method in section \ref{sec:LLT} further
to generate a family of bilinear relations involving the layer to layer transfer matrices 
with various boundary labels (theorem \ref{prop:S rel}).
They form the most general relations in this paper 
(see remark \ref{re:sar}), which ultimately specialize
to the hat relation.
In section \ref{sec:appli} we explain how the $q=0$ case of 
the results in section \ref{sec:fb}  yield the 
difference analogue of the hat relation (proposition \ref{pr:akn}).
The original 
hat relation is an immediate consequence of it as mentioned in the end.
Section \ref{sec:sum} is devoted to a summary and an outlook.

\section{$n$-species TASEP}\label{sec:tasep}
\subsection{Definition of $n$-TASEP}\label{subsec:def}
Consider the periodic 1D chain with $L$ sites  
$\Z_L$.
Each site $i \in \Z_L$ is populated with a local state
$\sigma_i \in \{0,1,\ldots, n\}$.
It is interpreted as the species of the particle occupying it or 
$0$ indicating the absence of particles. 
We assume $1 \le n < L$.
Consider a stochastic model on $\Z_L$ 
such that neighboring pairs of local states
$(\sigma_i, \sigma_{i+1})=(\alpha,\beta)$ 
are interchanged as $\alpha\, \beta \rightarrow \beta\, \alpha$
if $\alpha>\beta$ with the uniform transition rate.
The space of states is given by
\begin{align}\label{W}
(\C^{n+1})^{\otimes L} \simeq 
\bigoplus_{(\sigma_1,\ldots, \sigma_L) \in\{0,\ldots, n\}^L} 
\C|\sigma_1,\ldots, \sigma_L\rangle.
\end{align}
Let ${\mathbb P}(\sigma_1,\ldots, \sigma_L; t)$ be the probability of finding 
the configuration $(\sigma_1,\ldots, \sigma_L)$ at time $t$, and set 
\begin{align*}
|P(t)\rangle
= \sum_{(\sigma_1,\ldots, \sigma_L) \in\{0,\ldots, n\}^L}
{\mathbb P}(\sigma_1,\ldots, \sigma_L; t)|\sigma_1,\ldots, \sigma_L\rangle.
\end{align*}
By $n$-TASEP we mean the stochastic system 
governed by the continuous-time master equation
\begin{align*}
\frac{d}{dt}|P(t)\rangle
= H |P(t)\rangle,
\end{align*}
where the Markov matrix has the form
\begin{align}\label{H}
H = \sum_{i \in \Z_L} h_{i,i+1},\qquad
h |\alpha, \beta\rangle = \begin{cases}
|\beta, \alpha\rangle-|\alpha, \beta\rangle & 
\;(\alpha>\beta),\\
0 & \; (\alpha \le \beta).
\end{cases}
\end{align}
Here $h_{i,i+1}$ is the local Markov matrix that  
acts as $h$ on the $i$-th and the $(i+1)$-th components and 
as the identity elsewhere.
As $H$ preserves the particle content, 
it acts on each {\em sector}
consisting of the configurations with prescribed 
{\em multiplicity} ${\bf m}=(m_0,\ldots, m_n) \in (\Z_{\ge 0})^{n+1}$ of particles:
\begin{align*}
S({\bf m}) = 
\{{\boldsymbol \sigma}=(\sigma_1,\ldots, \sigma_L) \in \{0,\ldots, n\}^L\;|\;
\sum_{j=1}^L\delta_{k,\sigma_j}=m_k,\forall k\}. 
\end{align*}
The space of states (\ref{W}) is decomposed as 
$\bigoplus_{{\bf m}}
\bigoplus_{{\boldsymbol \sigma} \in S({\bf m}) }
\C | {\boldsymbol \sigma}\rangle$,
where the outer sum ranges over $m_i \in \Z_{\ge 0}$ such that 
$m_0+ \cdots + m_n = L$.
A sector $\bigoplus_{{\boldsymbol \sigma} \in S({\bf m}) }
\C | {\boldsymbol \sigma}\rangle$ such that $m_i \ge 1$ for all $0 \le i \le n$
is called {\em basic}.
Non-basic sectors are equivalent to a basic sector for $n'$-TASEP with some 
$n'<n$ by a suitable relabeling of species.
Thus we shall exclusively deal with basic sectors in this paper,
hence $n<L$ as mentioned before.
This condition guarantees  \cite{KMO} 
the convergence of the right hand side of (\ref{mho}).
The spectrum of $H$ is known to exhibit
a remarkable duality \cite{AKSS}.

\subsection{Steady states}\label{subsec:sst}
In each sector $\bigoplus_{{\boldsymbol \sigma} \in S({\bf m}) }
\C | {\boldsymbol \sigma}\rangle$ there is a unique vector  
$|{\bar P}({\bf m})\rangle$
up to a normalization, called the {\em steady state}, 
satisfying $H|{\bar P}({\bf m})\rangle = 0$.
The steady state for $1$-TASEP is trivial 
under the periodic boundary condition in that  
all the monomials have the same coefficient, i.e.
all the configurations are realized with an equal probability.

\begin{example}\label{ex:pbar}
We present (unnormalized) steady states in small sectors of 
$2$-TASEP and $3$-TASEP in the form
\begin{align*}
 |{\bar P}({\bf m})\rangle = |\xi({\bf m})\rangle + C|\xi({\bf m})\rangle
 + \cdots + C^{L-1}|\xi({\bf m})\rangle
 \end{align*}
respecting the symmetry $HC=CH$ under the 
 $\Z_L$ cyclic shift
 $C: |\sigma_1, \sigma_2,\ldots, \sigma_L\rangle \mapsto 
 |\sigma_L, \sigma_1, \ldots, \sigma_{L-1}\rangle$.
 The choice of the vector $|\xi({\bf m})\rangle$ is not unique.
 \begin{align*}
 |\xi(1,1,1)\rangle & =  2 | 012\rangle  + | 102\rangle,\\
 |\xi(2,1,1)\rangle  &= 3 | 0012\rangle  + 2 | 0102\rangle  + | 1002\rangle,\\
 |\xi(1,2,1)\rangle  &=  2 | 0112\rangle  + | 1012\rangle  + | 1102\rangle,\\
 |\xi(1,1,2)\rangle  &= 3 | 1220\rangle  + 2 | 2120\rangle  + | 2210\rangle,\\
 |\xi(1,2,2)\rangle & =3 | 11220\rangle + 2 | 12120\rangle + | 12210\rangle + 
 2 | 21120\rangle + | 21210\rangle + | 22110\rangle,\\
 |\xi(2,1,2)\rangle & =| 00221\rangle + 2 | 02021\rangle + 3 | 02201\rangle + 
 3 | 20021\rangle + 5 | 20201\rangle + 6 | 22001\rangle,\\
 |\xi(2,2,1)\rangle & =3 | 00112\rangle + 2 | 01012\rangle + 2 | 01102\rangle + 
 | 10012\rangle + | 10102\rangle + | 11002\rangle,\\
 |\xi(1,1,1,1)\rangle & =  9 | 0123\rangle  + 3 | 0213\rangle  + 3 | 1023\rangle  + 
 5 | 1203\rangle  + 3 | 2013\rangle  + | 2103\rangle,\\
 |\xi(2,1,1,1)\rangle & =  24 | 00123\rangle  + 6 | 00213\rangle  + 
 12 | 01023\rangle  + 17 | 01203\rangle  + 
 8 | 02013\rangle  + 3 | 02103\rangle\\  
 &+ 4 | 10023\rangle  + 
 7 | 10203\rangle  + 9 | 12003\rangle  + 6 | 20013\rangle  + 
 3 | 20103\rangle  + | 21003\rangle,\\
 |\xi(1,2,1,1)\rangle & =  12 | 01123\rangle  + 5 | 01213\rangle  + 3 | 02113\rangle  + 
 4 | 10123\rangle  + 3 | 10213\rangle  + 4 | 11023\rangle \\ 
 &+ 7 | 11203\rangle  + 5 | 12013\rangle  + 2 | 12103\rangle  + 
 3 | 20113\rangle  + | 21013\rangle  + | 21103\rangle,\\
 |\xi(1,1,2,1)\rangle & =  12 | 01223\rangle  + 5 | 02123\rangle  + 3 | 02213\rangle  + 
 3 | 10223\rangle  + 5 | 12023\rangle  + 7 | 12203\rangle  \\
 &+ 4 | 20123\rangle  + 3 | 20213\rangle  + | 21023\rangle  + 
 2 | 21203\rangle  + 4 | 22013\rangle  + | 22103\rangle,\\
 |\xi(1,1,1,2)\rangle & =  24 | 12330\rangle  + 12 | 13230\rangle  + 
 4 | 13320\rangle  + 6 | 21330\rangle  + 8 | 23130\rangle  + 
 6 | 23310\rangle \\ 
 &+ 17 | 31230\rangle  + 
 7 | 31320\rangle  + 3 | 32130\rangle  + 3 | 32310\rangle  + 
 9 | 33120\rangle  + | 33210\rangle.
\end{align*}
\end{example}
As these coefficients indicate, steady states are nontrivial for $n\ge 2$.
We will demonstrate the 
3D integrability behind them which will ultimately be 
related to the tetrahedron equation.

\subsection{Matrix product formula}\label{subsec:mp}
Consider the steady state
\begin{align}\label{srb}
|{\bar P}({\bf m})\rangle = \sum_{{\boldsymbol \sigma} \in S({\bf m})}
{\mathbb P}({\boldsymbol \sigma}) | {\boldsymbol \sigma} \rangle
\end{align}
and postulate that the steady state probability ${\mathbb P}({\boldsymbol \sigma})$ 
is expressed 
in the matrix product form
\begin{align}\label{mho}
{\mathbb P}(\sigma_1,\ldots, \sigma_L) 
= \mathrm{Tr}(X_{\sigma_1}\cdots X_{\sigma_L})
\end{align}
in terms of some operators $X_0, \ldots, X_n$.
Introduce the notations for the matrix elements of the local Markov matrix (\ref{H}) and 
the associated product of $X_i$'s as
\begin{align}\label{hdef}
h|\alpha, \beta \rangle = \sum_{\gamma,\delta}h^{\gamma,\delta}_{\alpha,\beta}
|\gamma, \delta\rangle,\qquad
(hXX)_{\alpha, \beta} := \sum_{\gamma,\delta}h^{\alpha,\beta}_{\gamma,\delta}
X_\gamma X_\delta.
\end{align}
Then we have
\begin{align*}
H  |{\bar P}({\bf m})\rangle 
&= \sum_{i \in \Z_L}
\sum_{\boldsymbol \sigma \in S({\bf m})} {\mathbb P}(\ldots, \sigma_i, \sigma_{i+1},\ldots)h_{i,i+1}
|\ldots, \sigma_i, \sigma_{i+1},\ldots\rangle\\
&= \sum_{i \in \Z_L}\sum_{\boldsymbol \sigma \in S({\bf m})} \sum_{\sigma'_i, \sigma'_{i+1}}
\mathrm{Tr}(\cdots X_{\sigma_i}X_{\sigma_{i+1}}\cdots )
h^{\sigma'_i, \sigma'_{i+1}}_{\sigma_i, \sigma_{i+1}}
|\ldots, \sigma'_i, \sigma'_{i+1},\ldots\rangle\\
&= \sum_{\boldsymbol \sigma \in S({\bf m})}\sum_{i \in \Z_L}
\mathrm{Tr}(\cdots (hXX)_{\sigma_i, \sigma_{i+1}}\cdots )
|\ldots, \sigma_i, \sigma_{i+1},\ldots\rangle.
\end{align*}
Therefore if there are another set of operators 
${\hat X}_0, \ldots, {\hat X}_n$ obeying the {\em hat relation}
\begin{align}\label{hrel}
(hXX)_{\alpha, \beta} = X_\alpha {\hat X}_\beta - {\hat X}_\alpha X_\beta,
\end{align}
the vector (\ref{srb})  satisfies $H|{\bar P}({\bf m})\rangle = 0$ 
thanks to the cyclicity of the trace (cf. \cite{BE}).
Then (\ref{mho}), if finite, must coincide with the actual steady state probability 
up to an overall normalization due to the uniqueness of the steady state.
Note on the other hand that ${\hat X}_i$ satisfying the hat relation 
with a given $X_i$ is not unique.
For instance 
${\hat X}_i \rightarrow {\hat X}_i+ c X_i$ leaves (\ref{hrel}) unchanged.

\subsection{Main result}\label{subsec:mr}
In our previous work \cite{KMO}, a new matrix product formula (\ref{mho}) of 
the steady state probability of the $n$-TASEP was proved 
which involves the operators $X_0, \ldots, X_n$ in the left diagram of 
\begin{equation}\label{ngm}
\begin{picture}(300,77)(-27,-10)

\put(-50,29){${X}_i=\sum$}
\put(20,52){$. . .$}
\put(-5,27){$.$}\put(-5,24){$.$}\put(-5,21){$.$}
\put(-8,48){\line(1,0){56}}
\put(-8,40){\line(1,0){48}}
\put(-8,32){\line(1,0){40}}
\put(-8,16){\line(1,0){24}}
\put(-8,8){\line(1,0){16}}
\put(-8,0){\line(1,0){8}}

\put(11,9.5){\put(29,25){$.$}\put(27,23){$.$}\put(25,21){$.$}}
\put(-9,-9.5){\put(29,25){$.$}\put(27,23){$.$}\put(25,21){$.$}}
\put(48,48){\vector(0,1){8}}
\put(40,40){\vector(0,1){16}}
\put(32,32){\vector(0,1){24}}
\put(16,16){\vector(0,1){40}}
\put(8,8){\vector(0,1){48}}
\put(0,0){\vector(0,1){56}}

\put(51,46){$\scriptstyle{0}$}
\put(43,38){$\scriptstyle{0}$}
\put(30,25){$\scriptstyle{0}$}

\put(24,18){$\scriptstyle{1}$}
\put(9,2){$\scriptstyle{1}$}
\put(0,-6){$\scriptstyle{1}$}

\put(29,52){\rotatebox{-135}{$\overbrace{\phantom{KKKK}}$}}
\put(54,26){$\scriptstyle{n-i}$}

\put(2,23){\rotatebox{-135}{$\overbrace{\phantom{KKKk}}$}}
\put(26,-2){$\scriptstyle{i}$}

\put(220,0){
\put(-117,29){${\hat X}_i=\sum (\alpha_1+\cdots+ \alpha_n)$}
\put(20,52){$. . .$}
\put(-5,27){$.$}\put(-5,24){$.$}\put(-5,21){$.$}
\put(-8,48){\line(1,0){56}}
\put(-8,40){\line(1,0){48}}
\put(-8,32){\line(1,0){40}}
\put(-8,16){\line(1,0){24}}
\put(-8,8){\line(1,0){16}}
\put(-8,0){\line(1,0){8}}

\put(11,9.5){\put(29,25){$.$}\put(27,23){$.$}\put(25,21){$.$}}
\put(-9,-9.5){\put(29,25){$.$}\put(27,23){$.$}\put(25,21){$.$}}
\put(48,48){\vector(0,1){8}}\put(44,60){$\scriptstyle{\alpha_n}$}
\put(40,40){\vector(0,1){16}}
\put(32,32){\vector(0,1){24}}
\put(16,16){\vector(0,1){40}}
\put(8,8){\vector(0,1){48}}\put(5,60){$\scriptstyle{\alpha_2}$}
\put(0,0){\vector(0,1){56}}\put(-7,60){$\scriptstyle{\alpha_1}$}

\put(51,46){$\scriptstyle{0}$}
\put(43,38){$\scriptstyle{0}$}
\put(30,25){$\scriptstyle{0}$}

\put(24,18){$\scriptstyle{1}$}
\put(9,2){$\scriptstyle{1}$}
\put(0,-6){$\scriptstyle{1}$}

\put(29,52){\rotatebox{-135}{$\overbrace{\phantom{KKKK}}$}}
\put(54,26){$\scriptstyle{n-i}$}

\put(2,23){\rotatebox{-135}{$\overbrace{\phantom{KKKk}}$}}
\put(26,-2){$\scriptstyle{i}$}
}
\end{picture}
\end{equation}
The proof was done by identifying the 
Ferrari-Martin algorithm \cite{FM} with a composition of the combinatorial $R$.
It did not rely on the hat relation, although 
${\hat X}_i$ defined by the right diagram 
was announced to fulfill it.
The main result of this paper is a self-contained proof of the 
hat relation (\ref{hrel}) which reads explicitly as follows:
\begin{theorem}[Hat relation]\label{th:mho}
The operators $X_i$ and ${\hat X}_i$ in (\ref{ngm}) satisfy
\begin{align*}
&\lbrack X_i,{\hat X}_j\rbrack=\lbrack {\hat X}_i,X_j \rbrack    
\quad\;\; \qquad (0 \le i,j \le n),\\
&X_iX_j={\hat X}_iX_j-X_i{\hat X}_j  \qquad (0 \le j < i \le n).
\end{align*}
\end{theorem}
The proof will be achieved in the end of section \ref{sec:appli}
as a consequence of its far-reaching generalization 
by embarking on $q$-deformed counterparts, layer to layer transfer matrices,
their bilinear relations and so forth.

In the rest of the section we explain the definition (\ref{ngm}).
First we consider the $X_i$ in the left diagram.
It represents a configuration sum, i.e. 
the partition function of the {\em $q=0$-oscillator valued five-vertex model}
on the triangular shape region of a square lattice with a prescribed 
condition along the SW-NE boundary.
\begin{equation} \label{5v0}
\begin{picture}(200,58)(-10,-36)
\thicklines

\put(-10,0){\vector(1,0){20}}
\put(0,-10){\vector(0,1){20}}
\put(-17,-3.5){0}\put(12.5,-3.5){0}\put(-2.4,13){0}\put(-2.3,-19){0}
\put(-2,-37){1}

\put(50,0){
\put(-10,0){\color{red} \vector(1,0){20}}
\put(0,-10){\color{red} \vector(0,1){20}}
\put(-17,-3.5){1}\put(12.5,-3.5){1}\put(-2.4,13){1}\put(-2.3,-19){1}
\put(-2,-37){1}}

\put(100,0){
\put(-10,0){\color{red}\line(1,0){10}}\put(0,0){\vector(1,0){10}}
\put(0,-10){\line(0,1){10}}\put(0,0){\color{red}\vector(0,1){10}}
\put(-17,-3.5){1}\put(12.5,-3.5){0}\put(-2.4,13){1}\put(-2.3,-19){0}
\put(-2.5,-37){${\bf a}^+$}}

\put(150,0){
\put(-10,0){\line(1,0){10}}\put(0,0){\color{red}\vector(1,0){10}}
\put(0,-10){\color{red}\line(0,1){10}}\put(0,0){\vector(0,1){10}}
\put(-17,-3.5){0}\put(12.5,-3.5){1}\put(-2.4,13){0}\put(-2.3,-19){1}
\put(-1.5,-37){${\bf a}^-$}}

\put(200,0){
\put(-10,0){\vector(1,0){20}}
\put(0,-10){\color{red}\vector(0,1){20}}
\put(-17,-3.5){0}\put(12.5,-3.5){0}\put(-2.4,13){1}\put(-2.3,-19){1}
\put(-1,-37){${\bf k}$}}
\thinlines
\end{picture}
\end{equation}
Each edge takes $0$ or $1$ and the sum extends over 
all the configurations such that every vertex 
is one of the above five types.
In (\ref{5v0}) we have colored the edges 
assuming $0$ and $1$ in black and red respectively.
This convention will apply in the rest of the paper\footnote{Although, in 
the formula like (\ref{ngm}),  the black edges not on the SW-NE boundary 
should be understood as taking both $0$ and $1$.}.
Given such a configuration, the summand is the 
{\em tensor product} of the local  ``Boltzmann weight"
 $1, {\bf a}^+,{\bf a}^-,{\bf k}$ assigned to each vertex as specified
 in the above\footnote{At 
the boundary corners in (\ref{ngm}) where arrows make $90^\circ$ left turns, 
we assume no change in the edge states and assign the weight 1.  
See examples \ref{ex:mrn2}  and \ref{ex:mrn3}.}.
They are linear operators 
on the Fock space $F= \bigoplus_{m \ge 0}\C |m\rangle$\footnote{The ket vector here 
should not be confused with the TASEP states in
section \ref{subsec:def}-\ref{subsec:mr}.} as
 ($|\!-\!1\rangle=0,\; 1 |m \rangle = |m\rangle$)
\begin{equation}\label{lil}
{\bf a}^+{\left| m\right\rangle}={\left| m+1\right\rangle}, \ \ {\bf a}^-{\left| m\right\rangle}={\left| m-1\right\rangle}, \ \ {\bf k}{\left| m\right\rangle}={\delta}_{m,0}{\left| m\right\rangle}
\end{equation}
obeying the relations
\begin{equation}\label{rna}
{\bf k} \ {\bf a}^{+}=0,\ \ \ \ \  \ {\bf a}^{-} \ {\bf k}=0,\ \ \ \ \ 
{\bf a}^+{\bf a}^-=1-{\bf k},\ \ \ \ \ {\bf a}^-{\bf a}^+=1.
\end{equation}
They are identified with the 
specialization of the $q$-oscillator algebra ${\mathscr A}_q$ in 
(\ref{tgm}) and (\ref{kyk}) to $q=0$.
Thus we write ${\bf a}^+,{\bf a}^-,{\bf k} \in {\mathscr A}_{q=0}$ here.
The $0$-oscillator operators attached to vertices at different positions act on  
independent copies of the Fock space.
Thus $X_i \in ({\mathscr A}_{q=0})^{\otimes n(n-1)/2} \subseteq 
\mathrm{End}(F^{\otimes n(n-1)/2})$.
Accordingly the trace in (\ref{mho}) is taken over $F^{\otimes n(n-1)/2}$.
In each component it is calculated by 
$\mathrm{Tr}_F(X) = \sum_{m \ge 0}\langle m | X | m \rangle$
with $\langle m | m'\rangle = \delta_{m,m'}$.

\begin{remark}\label{re:kwgc}
Our result (\ref{mho}) with (\ref{ngm}) 
corresponds to the {\em integer normalization}
\begin{align*}
{\mathbb P}(\sigma_1, \ldots, \sigma_L) = 1\quad \text{for}\;\;
\sigma_1 \ge \cdots \ge \sigma_L.
\end{align*}
In this normalization ${\mathbb P}({\boldsymbol \sigma}) \in 
\Z_{\ge 1}$ holds for all the state ${\boldsymbol \sigma} \in S({\bf m})$.
These facts can be shown via the equivalent formula \cite[eq.(4.3)]{KMO}
and the NY-rule for the combinatorial $R$ explained in 
\cite[sec.2.4]{KMO}. 
Example \ref{ex:pbar} has been given in this normalization.
\end{remark}

The $X_i$ has the form of a {\em corner transfer matrix} \cite{Bax} 
of the $0$-oscillator valued five-vertex model, although 
it acts along the perpendicular direction to the layer 
as opposed to the usual 2D setting.
Equivalently one may view it  
as a layer to layer transfer matrix of the 3D lattice model
where $F$ is assigned with the edges perpendicular 
to the plane on which the five-vertex model is defined.
The steady state probability (\ref{mho}) is then interpreted as a 
{\em partition function} of the 3D system of prism shape
which is periodic along the third direction.

As for the ${\hat X}_i$ in the right diagram of (\ref{ngm}), it means a similar 
configuration sum but now weighted by the coefficient
$\alpha_1+\cdots + \alpha_n$.

\begin{example}\label{ex:mrn2}
For $n=2$ the operator $X_i$ is given by
\begin{equation*}
\begin{picture}(600,40)(-25,10)
\thicklines
\setlength\unitlength{0.26mm}

\put(-25,37){$X_0= $}
\put(0,0){
\put(30,40){\vector(0,1){15}}
\put(50,40){\vector(0,1){15}}
\put(15,40){\line(1,0){15}}
\put(30,40){\line(1,0){20}}
\put(30,20){\line(0,1){20}}
\put(15,20){\line(1,0){15}}
}

\put(65,0){
\put(-5,37){$+$}
\put(30,40){\color{red}\vector(0,1){15}}
\put(50,40){\vector(0,1){15}}
\put(15,40){\color{red}\line(1,0){15}}
\put(30,40){\line(1,0){20}}
\put(30,20){\line(0,1){20}}
\put(15,20){\line(1,0){15}}
\put(60,37){$= 1 + {\bf a}^+,$}
}

\put(225,0){
\put(-25,37){$X_1= $}
\put(0,0){
\put(30,40){\color{red}\vector(0,1){15}}
\put(50,40){\vector(0,1){15}}
\put(15,40){\line(1,0){15}}
\put(30,40){\line(1,0){20}}
\put(30,20){\color{red}\line(0,1){20}}
\put(15,20){\color{red}\line(1,0){15}}
\put(60,37){$= {\bf k},$}
}}

\put(350,0){
\put(-25,37){$X_2= $}
\put(0,0){
\put(30,40){\vector(0,1){15}}
\put(50,40){\color{red}\vector(0,1){15}}
\put(15,40){\line(1,0){15}}
\put(30,40){\color{red}\line(1,0){20}}
\put(30,20){\color{red}\line(0,1){20}}
\put(15,20){\color{red}\line(1,0){15}}
}

\put(65,0){
\put(-5,37){$+$}
\put(30,40){\color{red}\vector(0,1){15}}
\put(50,40){\color{red}\vector(0,1){15}}
\put(15,40){\color{red}\line(1,0){15}}
\put(30,40){\color{red}\line(1,0){20}}
\put(30,20){\color{red}\line(0,1){20}}
\put(15,20){\color{red}\line(1,0){15}}
\put(60,37){$= {\bf a}^-+1.$}
}
}

\thinlines
\end{picture}
\end{equation*}
Accordingly we have 
${\hat X}_0 = {\bf a}^+,  {\hat X}_1 = {\bf k},  {\hat X}_2 = {\bf a}^-+2$.
Thus for instance,
\begin{align*}
{\mathbb P}(20201) = \mathrm{Tr}(X_2X_0X_2X_0X_1)
=\mathrm{Tr}\bigl((1+{\bf a}^+)(1+{\bf a}^-)(1+{\bf a}^+)(1+{\bf a}^-){\bf k}\bigl)=5
\end{align*}
reproducing the second last term in $|\xi(2,1,2)\rangle$ in example \ref{ex:pbar}.
\end{example}

\begin{example}\label{ex:mrn3}
For $n=3$ the operator $X_i$ is given by
\begin{equation*}
\begin{picture}(600,60)(-59,-17)
\thicklines
\setlength\unitlength{0.26mm}

\put(-55,24){$X_0= $}
\put(0,40){\line(1,0){50}} \put(50,40){\vector(0,1){13}}
\put(0,20){\line(1,0){30}} \put(30,20){\vector(0,1){33}}
\put(0,0){\line(1,0){10}}
\put(10,0){\vector(0,1){53}}
\put(-18,-20){$=1\otimes 1 \otimes 1$}

\put(90,0){\put(-20,24){$+$}
\put(10,40){\line(1,0){40}} \put(50,40){\vector(0,1){13}}
\put(0,20){\line(1,0){30}} \put(30,20){\vector(0,1){33}}
\put(0,0){\line(1,0){10}}
\put(10,0){\line(0,1){43}}
\put(0,40){\color{red}\line(1,0){10}}
\put(10,40){\color{red}\vector(0,1){13}}
\put(-22,-20){$+\;\; {\bf a}^+\!\otimes 1 \otimes 1$}}

\put(185,0){\put(-25,24){$+$}
\put(0,40){\line(1,0){50}} \put(50,40){\vector(0,1){13}}
\put(0,20){\color{red}{\line(1,0){10}}} \put(10,20){\line(1,0){20}}
\put(30,20){\vector(0,1){33}}
\put(0,0){\line(1,0){10}}
\put(10,0){\line(0,1){20}}\put(10,20){\color{red}\vector(0,1){33}}
\put(-22,-20){$+\;\; {\bf k}\otimes {\bf a}^+\! \otimes 1$}}

\put(290,0){\put(-30,24){$+$}
\put(0,40){\line(1,0){10}} \put(10,40){\color{red}\line(1,0){20}}\put(30,40){\line(1,0){20}} 
\put(50,40){\vector(0,1){13}}\put(10,40){\vector(0,1){13}}
\put(10,20){\color{red}\line(0,1){20}}
\put(0,20){\color{red}\line(1,0){10}}
\put(10,20){\line(1,0){20}} \put(30,40){\color{red}\vector(0,1){13}}
\put(30,20){\line(0,1){20}}
\put(0,0){\line(1,0){10}}
\put(10,0){\line(0,1){20}}
\put(-29,-20){$+\;\; {\bf a}^-\!\otimes {\bf a}^+\! \otimes {\bf a}^+$}}

\put(400,0){\put(-25,24){$+$}
\put(0,20){\color{red}\line(1,0){10}}\put(10,20){\color{red}\vector(0,1){33}}
\put(0,40){\color{red}\line(1,0){30}}\put(30,40){\color{red}\vector(0,1){13}}
\put(30,40){\line(1,0){20}} \put(50,40){\vector(0,1){13}}
\put(10,20){\line(1,0){20}} \put(30,20){\line(0,1){20}}
\put(0,0){\line(1,0){10}}
\put(10,0){\line(0,1){20}}
\put(-25,-20){$+\;\; 1\otimes {\bf a}^+ \otimes {\bf a}^+$,}}

\end{picture}
\end{equation*}
\begin{equation*}
\begin{picture}(600,60)(-59,-20)
\setlength\unitlength{0.26mm}
\thicklines

\put(-55,24){$X_1= $}
\put(0,40){\line(1,0){50}} \put(50,40){\vector(0,1){13}}
\put(0,20){\line(1,0){30}} \put(30,20){\vector(0,1){33}}
\put(0,0){\color{red}\line(1,0){10}}
\put(10,0){\color{red}\vector(0,1){53}}
\put(-18,-20){$={\bf k} \otimes {\bf k} \otimes 1$}

\put(95,0){\put(-20,24){$+$}
\put(10,40){\color{red}\line(1,0){20}} \put(30,40){\line(1,0){20}}
\put(50,40){\vector(0,1){13}}\put(30,40){\color{red}\vector(0,1){13}}
\put(0,20){\line(1,0){30}} \put(30,20){\line(0,1){20}}
\put(0,0){\color{red}\line(1,0){10}}
\put(10,0){\color{red}\line(0,1){43}}
\put(0,40){\line(1,0){10}}
\put(10,40){\vector(0,1){13}}
\put(-22,-20){$+\;\; {\bf a}^-\!\otimes {\bf k} \otimes {\bf a}^+$}}

\put(200,0){\put(-20,24){$+$}
\put(50,40){\vector(0,1){13}}\put(30,40){\color{red}\vector(0,1){13}}
\put(0,40){\color{red}\line(1,0){30}}\put(30,40){\line(1,0){20}}
\put(30,20){\line(0,1){20}}
\put(0,20){\line(1,0){30}}
\put(10,0){\color{red}\vector(0,1){53}}
\put(0,0){\color{red}\line(1,0){10}}
\put(-22,-20){$+\;\; 1\otimes {\bf k} \otimes {\bf a}^+$,}}

\end{picture}
\end{equation*}
\begin{equation*}
\begin{picture}(600,60)(-59,-20)
\setlength\unitlength{0.26mm}
\thicklines

\put(-55,24){$X_2= $}
\put(0,40){\line(1,0){50}}\put(50,40){\vector(0,1){13}}
\put(0,20){\line(1,0){10}}\put(10,20){\color{red}\line(1,0){20}}
\put(30,20){\color{red}\vector(0,1){33}}\put(10,20){\vector(0,1){33}}
\put(10,0){\color{red}\line(0,1){20}}
\put(0,0){\color{red}\line(1,0){10}}
\put(-22,-20){$\;=\;1\otimes {\bf a}^-\! \otimes {\bf k}$}

\put(105,0){\put(-28,24){$+$}
\put(0,40){\color{red}\line(1,0){10}}\put(10,40){\line(1,0){40}}
\put(50,40){\vector(0,1){13}}
\put(0,20){\line(1,0){10}}\put(10,20){\color{red}\line(1,0){20}}
\put(30,20){\color{red}\vector(0,1){33}}\put(10,40){\color{red}\vector(0,1){13}}
\put(10,20){\line(0,1){20}}
\put(10,0){\color{red}\line(0,1){20}}
\put(0,0){\color{red}\line(1,0){10}}
\put(-32,-20){$\;+\;\;{\bf a}^+\!\otimes{\bf a}^-\! \otimes {\bf k}$}
}

\put(205,0){\put(-25,24){$+$}
\put(0,40){\line(1,0){50}} \put(50,40){\vector(0,1){13}}
\put(0,20){\color{red}\line(1,0){30}} \put(30,20){\color{red}\vector(0,1){33}}
\put(0,0){\color{red}\line(1,0){10}}
\put(10,0){\color{red}\vector(0,1){53}}
\put(-21,-20){$+\;\;{\bf k}\otimes 1 \otimes {\bf k}$,}}

\end{picture}
\end{equation*}
\begin{equation*}
\begin{picture}(600,60)(-59,-20)
\setlength\unitlength{0.26mm}
\thicklines

\put(-55,24){$X_3= $}
\put(0,20){\line(1,0){10}}\put(10,20){\vector(0,1){33}}
\put(0,40){\line(1,0){30}}\put(30,40){\vector(0,1){13}}

\put(30,40){\color{red}\line(1,0){20}} \put(50,40){\color{red}\vector(0,1){13}}
\put(10,20){\color{red}\line(1,0){20}} \put(30,20){\color{red}\line(0,1){20}}
\put(0,0){\color{red}\line(1,0){10}}
\put(10,0){\color{red}\line(0,1){20}}
\put(-19,-20){$=1\otimes {\bf a}^-\! \otimes {\bf a}^-$}

\put(105,0){\put(-30,24){$+$}
\put(0,40){\color{red}\line(1,0){10}} \put(10,40){\line(1,0){20}}
\put(30,40){\color{red}\line(1,0){20}} 
\put(50,40){\color{red}\vector(0,1){13}}\put(10,40){\color{red}\vector(0,1){13}}
\put(10,20){\line(0,1){20}}
\put(0,20){\line(1,0){10}}
\put(10,20){\color{red}\line(1,0){20}} \put(30,40){\vector(0,1){13}}
\put(30,20){\color{red}\line(0,1){20}}
\put(0,0){\color{red}\line(1,0){10}}
\put(10,0){\color{red}\line(0,1){20}}
\put(-29,-20){$+\;\; {\bf a}^+\!\otimes {\bf a}^-\! \otimes {\bf a}^-$}}

\put(218,0){
\put(-30,24){$+$}
\put(0,40){\line(1,0){30}} \put(50,40){\color{red}\vector(0,1){13}}
\put(0,20){\color{red}\line(1,0){30}} \put(30,40){\vector(0,1){13}}
\put(0,0){\color{red}\line(1,0){10}}\put(30,20){\color{red}\line(0,1){20}}
\put(10,0){\color{red}\vector(0,1){53}}\put(30,40){\color{red}\line(1,0){20}}
\put(-28,-20){$+\;{\bf k}\otimes 1 \otimes {\bf a}^-$}}

\put(307,0){\put(-22,24){$+$}
\put(10,40){\color{red}\line(1,0){40}} \put(50,40){\color{red}\vector(0,1){13}}
\put(0,20){\color{red}\line(1,0){30}} \put(30,20){\color{red}\vector(0,1){33}}
\put(0,0){\color{red}\line(1,0){10}}
\put(10,0){\color{red}\line(0,1){40}}
\put(0,40){\line(1,0){10}}
\put(10,40){\vector(0,1){13}}
\put(-22,-20){$+\;\; {\bf a}^-\!\otimes 1 \otimes 1$}}

\put(402,0){\put(-25,24){$+$}
\put(0,40){\color{red}\line(1,0){50}} \put(50,40){\color{red}\vector(0,1){13}}
\put(0,20){\color{red}\line(1,0){30}} \put(30,20){\color{red}\vector(0,1){33}}
\put(0,0){\color{red}\line(1,0){10}}
\put(10,0){\color{red}\vector(0,1){53}}
\put(-21,-20){$+\;1\otimes 1\otimes 1$.}}

\end{picture}
\end{equation*}
Here and in what follows, 
the components of the tensor product will always be ordered 
so that they correspond, from left to right, 
to the vertices (if exist) at 
$(1,1)$, $(2,1)$, $(1,2)$, $(3,1)$, $(2,2)$, $(1,3), \ldots$, 
where $(i,j)$ is the intersection of the $i$-th horizontal line from the top
and the $j$-th vertical line from the left.
Accordingly we have
\begin{align*}
{\hat X}_0 &= {\bf a}^+ \ot 1 \ot 1 + {\bf k} \ot {\bf a}^+ \ot 1 
+ {\bf a}^-\ot {\bf a}^+ \ot {\bf a}^+ + 2(1\ot {\bf a}^+ \ot {\bf a}^+),\\
{\hat X}_1 &= {\bf k} \ot {\bf k} \ot 1 + {\bf a}^- \ot {\bf k} \ot {\bf a}^+ 
+ 2(1\ot {\bf k} \ot {\bf a}^+),\\
{\hat X}_2 &= 1 \ot {\bf a}^- \ot {\bf k} + 2{\bf a}^+ \ot {\bf a}^- \ot {\bf k}
+ 2 {\bf k} \ot 1 \ot {\bf k},\\
{\hat X}_3 &= 1 \ot {\bf a}^- \ot {\bf a}^- + 
2 {\bf a}^+ \ot {\bf a}^- \ot {\bf a}^- + 2 {\bf k} \ot 1 \ot {\bf a}^- +
2{\bf a}^- \ot 1 \ot 1 + 3(1 \ot 1 \ot 1).
\end{align*}
\end{example}

\section{3D $L,M$ operators and the tetrahedron equation}\label{sec:LM}
In this section, we define $L$ and $M$ operators 
involving a generic parameter $q$ and 
describe their properties used in later sections. 

\subsection{$q$-oscillator algebra and the Fock space} 
Let $q$ be a generic complex parameter unless it is set to be $0$
in section \ref{sec:appli}.
Let ${\mathscr A}_q$ be the $q$-oscillator algebra generated 
by ${\bf a}^{+},{\bf a}^-,{\bf k}$ with relations
\begin{equation}\label{tgm}
{\bf k} \, {\bf a}^{\pm}=-q^{\pm1}{\bf a}^{\pm}\,{\bf k},
\ \ \ \ \ {\bf a}^+\,{\bf a}^-=1-{\bf k}^2,\ \ \ \ \ {\bf a}^-\,{\bf a}^+=1-q^2{\bf k}^2.
\end{equation}
We also consider ${\tilde{\mathscr A}}_q={\mathscr A}_{-q}$ 
generated by ${\bf a}^{+},{\bf a}^-,{\tilde{\bf k}}$ with relations
\begin{equation*}
{\tilde {\bf k}} \, {\bf a}^{\pm}=q^{\pm1}{\bf a}^{\pm}\,{\tilde {\bf k}},\ \ \ \ \ 
{\bf a}^+\,{\bf a}^-=1-{\tilde {\bf k}}^2,\ \ \ \ \ 
{\bf a}^-\,{\bf a}^+=1-q^2{\tilde {\bf k}}^2.
\end{equation*}
${\mathscr A}_q$ and ${\tilde{\mathscr A}}_q$ act on the Fock space 
$F={\bigoplus}_{m{\ge}0}{\mathbb C}{\left| m\right\rangle}$ 
as\footnote{We warn that the same notation 
${\bf a}^\pm, {\bf k}$ and $F$ will be used either for $q=0$ or not.
} 
\begin{equation}\label{kyk}
{\bf a}^+{\left| m\right\rangle}={\left| m+1\right\rangle}, \ \ {\bf a}^-{\left| m\right\rangle}=(1-q^{2m}){\left| m-1\right\rangle}, \ \ {\bf k}{\left| m\right\rangle}=(-q)^m{\left| m\right\rangle}, \ \ {\tilde{\bf k}}{\left| m\right\rangle}=q^m{\left| m\right\rangle}.
\end{equation}
We define the dual Fock space 
$F^*={\bigoplus}_{m{\ge}0}{\mathbb C}{\left\langle m\right|}$ 
on which ${\mathscr A}_q$ and ${\tilde{\mathscr A}}_q$ act from right as
\begin{equation*}
{\left\langle m\right|}{\bf a}^+=(1-q^{2m}){\left\langle m-1\right|}, \ \ \ \ \ {\left\langle m\right|}{\bf a}^-={\left\langle m+1\right|}, \ \ \ \ \ {\left\langle m\right|}{\bf k}=(-q)^m{\left\langle m\right|}, \ \ \ \ \ {\left\langle m\right|}{\tilde {\bf k}}=q^m{\left\langle m\right|}.
\end{equation*}
The pairing $F^*\ot F\longrightarrow\C$ is determined as
$\left\langle m | m^{\prime} \right\rangle=(q^2)_{m}{\delta}_{m,m^{\prime}}$ 
with $(q)_m = \prod_{1 \le j \le m}(1-q^j)$ 
so as to satisfy $(\left\langle m\right| X){\left| m^{\prime} \right\rangle}
=\left\langle m\right| (X{\left| m^{\prime} \right\rangle})$ 
for any $X\in{\mathscr A}_q,{\tilde{\mathscr A}}_q$.

\vspace{0.1cm}
We finally prepare the two-dimensional vector space $V$ and its dual $V^*$ by
\begin{align}\label{V}
V = \C v_0 \oplus \C v_1,\quad
V^* = \C v^*_0 \oplus \C v^*_1,\quad
\langle v^*_i,v_j\rangle=\delta_{ij}.
\end{align}

\subsection{3D $L,M$ operators with spectral parameter} \label{subsec:L and M}
We introduce 3D $L,M$ operators \cite{BS} with spectral parameter $z$. 
They are linear operators on $V\ot V\ot F$. 
For $i,j\in\{0,1\}$ and $|\xi\rangle\in F$, define $\L(z)$ by
\begin{equation}\label{hrk}
\L(z)(v_i\ot v_j\ot{\left| \xi \right\rangle})=\sum_{a,b=0,1}v_a\ot v_b{\ot}\L(z)^{a,b}_{i,j}{\left| \xi \right\rangle},
\end{equation}
where $\L_{i,j}^{a,b}(z)$ is an operator on $F$ such that
\begin{equation}\label{air}
\L(z)^{0,0}_{0,0}=\L(z)^{1,1}_{1,1}=1, \ \ \L(z)^{0,1}_{1,0}=z{\bf a}^+, \ \ \L(z)^{1,0}_{0,1}=z^{-1}{\bf a}^-, \ \ \L(z)^{0,1}_{0,1}={\bf k}, \ \ \L(z)^{1,0}_{1,0}=q{\bf k}.
\end{equation}
The other $\L_{i,j}^{a,b}(z)$'s are set to be $0$. 
One can let $\L(z)$ act from right
on $V^*\ot V^*\ot F^*$ as 
\begin{equation*}
(v^*_a\ot v^*_b\ot \langle \xi|\, )\L(z)=\sum_{i,j=0,1}v^*_i\ot v^*_j{\ot}{\langle \xi|}\L(z)^{a,b}_{i,j}.
\end{equation*}
$\M(z)$ is defined similarly with
\begin{equation*}
\M(z)^{0,0}_{0,0}=\M(z)^{1,1}_{1,1}=1, \ \ \M(z)^{0,1}_{1,0}=z{\bf a}^+, \ \ \M(z)^{1,0}_{0,1}=z^{-1}{\bf a}^-, \ \ \M(z)^{0,1}_{0,1}={\tilde {\bf k}}, \ \ \M(z)^{1,0}_{1,0}=-q{\tilde {\bf k}}.
\end{equation*}
Remark that the $L$ operator in \cite[eq.(2.9)]{KMO} 
corresponds to the $z=1$ case of the present 
$L$ operator equipped with the spectral parameter $z$. 
Graphically they are expressed as follows:
\begin{equation}\label{6v}
\begin{picture}(200,115)(-10,-90)
\thicklines

\put(-60,0){
\put(-15,0){\vector(1,0){30}}
\put(0,-15){\vector(0,1){30}}
\put(-19,-3.5){$i$}\put(17.5,-3){$a$}\put(-2.4,18){$b$}\put(-2.3,-24){$j$}
}

\put(-70,-44){$\L(z)^{a,b}_{i,j}$}

\put(-10,0){\vector(1,0){20}}
\put(0,-10){\vector(0,1){20}}
\put(-17,-3.5){0}\put(12.5,-3.5){0}\put(-2.4,13){0}\put(-2.3,-19){0}
\put(-2,-44){1}

\put(50,0){
\put(-10,0){\color{red}\vector(1,0){20}}
\put(0,-10){\color{red}\vector(0,1){20}}
\put(-17,-3.5){1}\put(12.5,-3.5){1}\put(-2.4,13){1}\put(-2.3,-19){1}
\put(-2,-44){1}}

\put(100,0){
\put(-10,0){\color{red}\line(1,0){10}}\put(0,0){\vector(1,0){10}}
\put(0,-10){\line(0,1){10}}\put(0,0){\color{red}\vector(0,1){10}}
\put(-17,-3.5){1}\put(12.5,-3.5){0}\put(-2.4,13){1}\put(-2.3,-19){0}
\put(-5.5,-44){$z{\bf a}^+$}}

\put(150,0){
\put(-10,0){\color{black}\line(1,0){10}}\put(0,0){\color{red}\vector(1,0){10}}
\put(0,-10){\color{red}\line(0,1){10}}\put(0,0){\color{black}\vector(0,1){10}}
\put(-17,-3.5){0}\put(12.5,-3.5){1}\put(-2.4,13){0}\put(-2.3,-19){1}
\put(-10.5,-44){$z^{-1}{\bf a}^-$}}

\put(200,0){
\put(-10,0){\vector(1,0){20}}
\put(0,-10){\color{red}\vector(0,1){20}}
\put(-17,-3.5){0}\put(12.5,-3.5){0}\put(-2.4,13){1}\put(-2.3,-19){1}
\put(-2.5,-44){${\bf k}$}}

\put(250,0){
\put(-10,0){\color{red}\vector(1,0){20}}
\put(0,-10){\vector(0,1){20}}
\put(-17,-3.5){1}\put(12.5,-3.5){1}\put(-2.4,13){0}\put(-2.3,-19){0}
\put(-3.5,-44){$q{\bf k}$}}

\put(-70,-70){$\M(z)^{a,b}_{i,j}$}

\put(-2,-70){$1$}

\put(48,-70){$1$}

\put(94.5,-70){$z{\bf a}^+$}

\put(140,-70){$z^{-1}{\bf a}^-$}

\put(198,-70){${\tilde {\bf k}}$}

\put(242,-70){$-q{\tilde {\bf k}}$}
\thinlines
 \end{picture}
\end{equation}
Note that $z$ is not exhibited in the diagrams for simplicity. 
In view of the property
\begin{align}\label{ice}
\L(z)^{a,b}_{i,j} = \M(z)^{a,b}_{i,j} = 0 \;\;
\text{unless}\;\; a+b=i+j,
\end{align}
$\L(z)$ and $\M(z)$ can be considered to define  
$q$-oscillator valued six-vertex models on the 2D lattice. 
Alternatively, we can regard $\L(z)$ and $\M(z)$ 
as vertices on the 3D square lattice as
\begin{equation}\label{RLpic}
\begin{picture}(220,60)(-310,-30)
\thicklines
\put(-280,5){
\put(-75,-7){$\L(z)^{a,b}_{i,j}\;=$}

\rotatebox{20}{
{\linethickness{0.2mm}
\put(13,-4){\color{blue}\vector(-1,0){40}}}}

\put(-5,-18){\vector(0,1){32}}
\put(-20,0){\vector(3,-1){33}}

\put(-26,0){$\scriptstyle{i}$}
\put(-6,17){$\scriptstyle{b}$}
\put(-6,-26){$\scriptstyle{j}$}
\put(15,-14){$\scriptstyle{a}$}
}

\put(-100,5){
\put(-75,-7){$\M(z)^{a,b}_{i,j}\;=$}

\rotatebox{20}{
{\linethickness{0.2mm}
\put(13,-4){\color{green}\vector(-1,0){40}}}}

\put(-5,-18){\vector(0,1){32}}
\put(-20,0){\vector(3,-1){33}}

\put(-26,0){$\scriptstyle{i}$}
\put(-6,17){$\scriptstyle{b}$}
\put(-6,-26){$\scriptstyle{j}$}
\put(15,-14){$\scriptstyle{a}$}
}
\end{picture}
\thinlines
\end{equation}
Here, along the blue or green line runs the Fock space $F$. We use 
the two colors to distinguish $\L(z)$ from $\M(z)$.

\subsection{Right and left eigenvectors of the $M$ operator}
Let us provide some right and left eigenvectors of $\M(z)$ for later use. 
\begin{proposition} \label{prop:right eigen}
Set $\left| {\chi}(z) \right\rangle={\sum_{m \ge 0}\frac{z^m}{(q)_{m}}{\left| m \right\rangle}}$. Then the following vectors are right eigenvectors of $\M(z)$ with eigenvalue 1
for any ${\left\langle \xi \right|}{\in}F^*$ and ${\alpha},{\beta}\in{\mathbb C}$.
\begin{equation*}
v_0\ot v_0{\ot}{\left| \xi \right\rangle}, \ \ v_1\ot v_1{\ot}{\left| \xi \right\rangle}, 
\ \ ({\alpha}v_1\ot v_0+{\beta}v_0\ot v_1){\ot}
|\chi ({\textstyle\frac{\alpha z}{\beta}}) \rangle.
\end{equation*}
\end{proposition}
\begin{proof}
The first two are obvious. The last one is verified by directly checking 
\begin{align}\label{sin1}
\sum_{i+j=1}\alpha^i\beta^j\M(z)_{i,j}^{k,l}
|\chi({\textstyle\frac{\alpha z}{\beta}})\rangle
=\alpha^k\beta^l|\chi({\textstyle\frac{\alpha z}{\beta}})\rangle.
\end{align}
\end{proof}
Similarly, we have
\begin{proposition} \label{prop:left eigen}
Set $\left\langle {\chi}(z) \right|=\sum_{m \ge 0}\frac{z^m}{(q)_m}\langle m|$. 
Then the following vectors are left eigenvectors of $\M(z)$ with eigenvalue 1
for any ${\left\langle \xi \right|}{\in}F^*$ and ${\alpha},{\beta}\in{\mathbb C}$.
\begin{equation*}
v^*_0\ot v^*_0{\ot}{\left\langle \xi \right|}, 
\ \ v^*_1\ot v^*_1{\ot}{\left\langle \xi \right|}, 
\ \ (\alpha v^*_1\ot v^*_0+\beta v^*_0\ot v^*_1){\ot}\langle 
\chi({\textstyle\frac{\alpha}{\beta z}})|.
\end{equation*}
\end{proposition}
\begin{proof}
Again the first two are trivial and the last one is due to  
\begin{align}\label{sin2}
\sum_{i+j=1}\alpha^i\beta^j
\langle\chi({\textstyle \frac{\alpha}{\beta z}})|\M(z)^{i,j}_{k,l}
=\alpha^k\beta^l
\langle\chi({\textstyle \frac{\alpha}{\beta z}})|.
\end{align}
\end{proof}
The above propositions imply
\begin{corollary} \label{cor:eigen}
$(v_0+v_1)^{\ot2}{\ot}{\left| \chi(z) \right\rangle}$ 
$\mathrm{(}$resp. $(v_0^*+v_1^*)^{\ot2}{\ot}{\left\langle \chi(z^{-1}) 
\right|}$ $\mathrm{)}$ is also a right 
$\mathrm{(}$resp. left$\mathrm{)}$ 
eigenvector of $\M(z)$ of eigenvalue 1, i.e.,
\begin{equation}\label{mst}
\sum_{i,j}\M(z)^{k,l}_{i,j}|\chi(z)\rangle=|\chi(z)\rangle,\quad
\langle\chi(z^{-1})|\sum_{i,j}\M(z)^{i,j}_{k,l}=\langle\chi(z^{-1})|
\end{equation}
hold for any $k,l=0,1$.
\end{corollary}

\subsection{Tetrahedron equation}
The $L,M$ operators introduced in 
section \ref{subsec:L and M} satisfy the tetrahedron equation.
\begin{theorem} [Tetrahedron equation]\label{prop:tetrahedron}
As an operator on $V^{\ot4}\ot F^{\ot2}$ the following equality holds:
\begin{equation} \label{aoy}
\M_{126}(z_{12})\M_{346}(z_{34})\L_{135}(z_{13})\L_{245}(z_{24})
= \L_{245}(z_{24})\L_{135}(z_{13})\M_{346}(z_{34})\M_{126}(z_{12}),
\end{equation}
where $z_{ij}=z_i/z_j$. Graphically it looks as
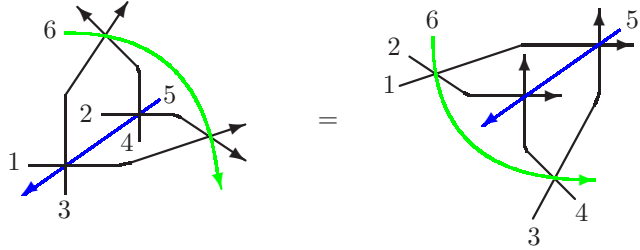
\begin{equation*}
\begin{picture}(400,90)(-80,-1)
\setlength{\unitlength}{0.32mm}
\thicklines

\put(59,50){$5$}
\rotatebox{35}{\put(70,14){\color{blue}\vector(-1,0){70}}}

\put(7,2){$3$}
\put(10,12){\line(0,1){38}}\qbezier(10,50)(10,51.5)(10,53)
\put(10,53){\vector(2,3){26}}

\put(33,30){$4$}
\put(41,34){\line(0,1){28}}\put(40,65){\vector(-1,1){25}}
\qbezier(41,62)(40.5,63.5)(40,65)

\put(16,42){$2$}
\put(25,45){\line(1,0){30}}\put(58,44){\vector(3,-2){27}}
\qbezier(55,45)(56.5,44.5)(58,44)

\put(-14,22){$1$}
\put(-5,24){\line(1,0){37}}
\qbezier(32,24)(34.5,24.5)(37,25)
\put(37,25){\vector(3,1){48}}

\put(1,77){$6$}
{\color{green}\qbezier(10,79)(67,80)(73.5,19)}
\put(74,18){\color{green}\vector(1,-4){1}}

\put(115,40){$=$}
\put(190,31){

\put(-48,23){$1$}
\put(-41,26){\line(3,1){51}}\put(13,43){\vector(1,0){43}}
\qbezier(10,43)(11.5,43)(13,43)

\put(-46,39){$2$}
\put(-14,23){\line(-3,2){23}}\put(-12,22){\vector(1,0){38}}
\qbezier(-14,23)(-13,22.5)(-12,22)

\put(12,-40){$3$}
\put(42,24){\vector(0,1){34}}
\qbezier(42,24)(41.5,21)(39.8,18)
{\rotatebox[origin=l]{-10}{\put(31.5,24){\line(-1,-3){17}}}}

\put(23,-30){$4$}
\put(-8,1){
\put(10,0){\vector(0,1){38}}\put(11,-2){\line(1,-1){20}}
\qbezier(10,0)(10.5,-1)(11,-2)}

\put(44,50){$5$}
\put(-20,-2){
\rotatebox{35}{\put(70,14){\color{blue}\vector(-1,0){70}}}}

\put(-39,50){$6$}
\put(-8,9){
{\color{green}\qbezier(-32,37)(-32,-23)(35,-22)}
\put(34.5,-22){\color{green}\vector(1,0){1}}}

}

\end{picture}
\thinlines
\end{equation*}
\end{theorem}
\begin{proof}
For instance, we have
\begin{align*}
&\langle v^*_0\ot v^*_1\ot v^*_0\ot v^*_1,
(\text{LHS})v_1\ot v_1\ot v_0\ot v_0\rangle \\
&=\L(z_{13})^{10}_{10}\L(z_{24})^{01}_{10}{\ot}\M(z_{12})^{01}_{10}\M(z_{34})^{01}_{01}+\L(z_{13})^{01}_{10}\L(z_{24})^{10}_{10}{\ot}\M(z_{12})^{01}_{01}\M(z_{34})^{01}_{10} \\ 
&=(q{\bf k}\cdot z_{24}{\bf a}^+){\ot}(z_{12}{\bf a}^+\cdot{\tilde {\bf k}}) +(z_{13}{\bf a}^+\cdot q{\bf k}){\ot}({\tilde {\bf k}}\cdot z_{34}{\bf a}^+)=0
\end{align*}
on $F^{\ot2}$, where the pairing is evaluated between $(V^*)^{\ot4}$ and $V^{\ot4}$.
On the other hand, one clearly has $\langle v^*_0\ot v^*_1\ot v^*_0\ot v^*_1,
(\text{RHS})v_1\ot v_1\ot v_0\ot v_0\rangle=0$.
The other cases can be shown similarly.
\end{proof}

The above type of the tetrahedron equation was first considered in \cite{BS}. 
In fact, our solutions $\L(z),\M(z)$ of the tetrahedron equation are
equivalent to that in \cite{BS} with a 
certain specialization of their spectral parameters,
up to a gauge transformation of the form 
\[
\L(z)\longrightarrow P_1(\alpha)P_2(\beta)\L(z)P_1(\alpha')^{-1}P_2(\beta')^{-1}
\]
and similarly for $\M(z)$, where $P_i(\gamma)$ 
acts nontrivially only on the $i$-th $V$.
We will see that the tetrahedron equation (\ref{aoy}) plays the most 
fundamental role controlling the whole family of relations 
among layer to layer transfer matrices and
ultimately the hat relation in theorem \ref{th:mho}.

\section{Layer to layer transfer matrix}\label{sec:LLT}
Here we study the partition functions of the 
$q$-oscillator valued six-vertex model
with special boundary conditions.
Put in another way, they are layer to layer transfer matrices
of a 3D lattice model whose basic unit is the 3D $L$ operator.

Fixing positive integers $m,n$, we define a linear operator $T(z)$ on 
$V^{\ot m}{\ot}{V^{\ot n}}{\ot}{F^{\ot mn}}$ graphically as follows:
\begin{equation*}
\begin{picture}(150,77)(-50,-10)

\put(-49,29){$T(z)= $}
\put(20,52){$. . .$}
\put(-5,27){$.$}\put(-5,24){$.$}\put(-5,21){$.$}
\put(-8,48){\vector(1,0){56}}
\put(-8,40){\vector(1,0){56}}
\put(-8,32){\vector(1,0){56}}
\put(-8,16){\vector(1,0){56}}
\put(-8,8){\vector(1,0){56}}
\put(-8,0){\vector(1,0){56}}

\put(40,-5){\vector(0,1){60}}
\put(32,-5){\vector(0,1){60}}
\put(16,-5){\vector(0,1){60}}
\put(8,-5){\vector(0,1){60}}
\put(0,-5){\vector(0,1){60}}

\put(-2,50){\rotatebox{-0}{$\overbrace{\phantom{KKKKK}}$}}
\put(18,68){$\scriptstyle{n}$}

\put(45,50){\rotatebox{-90}{$\overbrace{\phantom{KKKKKk}}$}}
\put(60,22){$\scriptstyle{m}$}

\end{picture}
\end{equation*}
Each line, horizontal or vertical, carries $V$ (\ref{V}). 
Each vertex represents $\L(z)^{a,b}_{i,j}$ in (\ref{6v})
including the spectral parameter $z$.
Penetrating each vertex from back to face, the Fock space
$F$ runs along a blue line as in the left figure in (\ref{RLpic}). 
When this feature is to be emphasized, we depict $T(z)$, say for $(m,n)=(3,4)$, as 
\begin{equation*}
\begin{picture}(100,50)(-15,5)
\setlength{\unitlength}{0.5mm}

\put(-28,20){$T(z)=$}
\put(0,30){\rotatebox{-5}{\vector(1,0){47}}}
\put(0,20){\rotatebox{-5}{\vector(1,0){47}}}
\put(0,10){\rotatebox{-5}{\vector(1,0){47}}}

\multiput(-1,2)(10,-1){4}{
\put(10,0){\vector(0,1){35}}}

\multiput(10,10)(10,-1){4}{
\multiput(0,0)(0,10){3}{
\put(2.4,1.6){\color{blue}\vector(-3,-2){9}}
}}

\end{picture}
\end{equation*}

Introduce the following notation:
\begin{align*}
&|{\bf i}\rangle=v_{i_1}\ot v_{i_2}\ot\cdots\ot v_{i_m},
&|{\bf j}\rangle=v_{j_1}\ot v_{j_2}\ot\cdots\ot v_{j_n},\\
&\langle{\bf a}|=v^*_{a_1}\ot v^*_{a_2}\ot\cdots\ot v^*_{a_m},
&\langle{\bf b}|=v^*_{b_1}\ot v^*_{b_2}\ot\cdots\ot v^*_{b_n},
\end{align*}
where all subscripts $i_1,i_2,\ldots$, etc are $0$ or $1$. Then
$T(z)^{{\bf a}, {\bf b}}_{{\bf i}, {\bf j}}=
(\langle{\bf a}|\ot\langle{\bf b}|)T(z)(|{\bf i}\rangle\ot|{\bf j}\rangle)\in{\rm End}
(F^{\ot mn})$ is represented as
\begin{equation}
\begin{picture}(130,77)(-50,-15)
\put(-69,29){$T(z)^{{\bf a}, {\bf b}}_{{\bf i}, {\bf j}}=$}
\put(-8,48){\vector(1,0){56}}
\put(-8,40){\vector(1,0){56}}
\put(-8,32){\vector(1,0){56}}
\put(-8,8){\vector(1,0){56}}
\put(-8,0){\vector(1,0){56}}

\put(40,-5){\vector(0,1){60}}
\put(32,-5){\vector(0,1){60}}
\put(8,-5){\vector(0,1){60}}
\put(0,-5){\vector(0,1){60}}

\put(-5,60){${\scriptstyle b_1}$}
\put(5,60){${\scriptstyle b_2}$}
\put(20,60){$. . .$}
\put(38,60){${\scriptstyle b_n}$}

\put(-5,-12){${\scriptstyle j_1}$}
\put(5,-12){${\scriptstyle j_2}$}
\put(20,-12){$. . .$}
\put(38,-12){${\scriptstyle j_n}$}

\put(50,48){${\scriptstyle a_1}$}
\put(50,38){${\scriptstyle a_2}$}
\put(50,14){$\vdots$}
\put(50,-2){${\scriptstyle a_m}$}

\put(-17,48){${\scriptstyle i_1}$}
\put(-17,38){${\scriptstyle i_2}$}
\put(-17,14){$\vdots$}
\put(-17,-2){${\scriptstyle i_m}$}

\nonumber
\end{picture}
\end{equation}
where the sums are taken over $\{0,1\}$ for all the internal edges.
With this notation, fixing $\langle{\bf a}|,|{\bf j}\rangle$ we set
\begin{equation} \label{S}
\begin{picture}(130,83)(-80,-10)
\put(-135,29){$S(z)^{\bf a}_{\bf j}
={\displaystyle \sum_{{\bf i},{\bf b}}}\,T(z)^{{\bf a}, {\bf b}}_{{\bf i}, {\bf j}}= $}
\put(-35,29){${\displaystyle {\sum_{{\bf i},{\bf b}}^{}}}$}
\put(-8,48){\vector(1,0){56}}
\put(-8,40){\vector(1,0){56}}
\put(-8,32){\vector(1,0){56}}
\put(-8,8){\vector(1,0){56}}
\put(-8,0){\vector(1,0){56}}

\put(40,-5){\vector(0,1){60}}
\put(32,-5){\vector(0,1){60}}
\put(8,-5){\vector(0,1){60}}
\put(0,-5){\vector(0,1){60}}

\put(-5,60){${\scriptstyle b_1}$}
\put(5,60){${\scriptstyle b_2}$}
\put(20,60){$. . .$}
\put(38,60){${\scriptstyle b_n}$}

\put(-5,-12){${\scriptstyle j_1}$}
\put(5,-12){${\scriptstyle j_2}$}
\put(20,-12){$. . .$}
\put(38,-12){${\scriptstyle j_n}$}

\put(50,48){${\scriptstyle a_1}$}
\put(50,38){${\scriptstyle a_2}$}
\put(50,14){$\vdots$}
\put(50,-2){${\scriptstyle a_m}$}

\put(-17,48){${\scriptstyle i_1}$}
\put(-17,38){${\scriptstyle i_2}$}
\put(-17,14){$\vdots$}
\put(-17,-2){${\scriptstyle i_m}$}

\put(70,29){$\in \ {\rm End}({F^{\ot mn}})$.}
\end{picture}
\end{equation}
The operators 
$T(z), T(z)^{{\bf a}, {\bf b}}_{{\bf i}, {\bf j}}$ and 
$S(z)^{\bf a}_{\bf j}$
are the layer to layer transfer matrices of size $m \times n$ with 
free, fixed and mixed (NW-free and SE-fixed) boundary conditions, respectively.

\begin{example}\label{ex:ts0}
Consider the case $(m,n)=(1,1)$.  
Then we have $T(z)^{a,b}_{i,j} = \L(z)^{a,b}_{i,j}$, therefore
\begin{align*}
S(z)^0_0 = 1+z{\bf a}^+,\quad
S(z)^1_1 = 1+z^{-1}{\bf a}^-,\quad
S(z)^0_1 = {\bf k},\quad
S(z)^1_0 = q {\bf k}.
\end{align*}
\end{example}

\begin{example}\label{ex:nmi}
Consider the case $(m,n)=(1,2)$. 
We list those $S(z)^{\bf a}_{\bf j}$ which will be used 
in example \ref{ex:ykw}.
\begin{equation*}
\!\!\!\!\!\!\!\!\!\!\!\!\!\!
\!\!\!\!\!\!\!\!\!\!\!\!\!\!\!\!\!\!
S^0_{00}(z) =
\begin{picture}(150,20)(0,0)
\thicklines
\mbox{\begin{picture}(40,18)(-6,12)
\put(6,14){\vector(0,1){10}}
\put(18,14){\vector(0,1){10}}
\put(-2,14){\line(1,0){8}}
\put(6,14){\line(1,0){12}}
\put(18,14){\vector(1,0){11}}
\put(6,7){\line(0,1){7}}
\put(18,7){\line(0,1){7}}
\end{picture}}
+
\mbox{\begin{picture}(40,18)(-6,12)
\put(6,14){\color{red}\vector(0,1){10}}
\put(18,14){\vector(0,1){10}}
\put(-2,14){\color{red}\line(1,0){8}}
\put(6,14){\line(1,0){12}}
\put(18,14){\vector(1,0){11}}
\put(6,7){\line(0,1){7}}
\put(18,7){\line(0,1){7}}
\end{picture}}
+
\mbox{\begin{picture}(40,18)(-6,12)
\put(6,14){\vector(0,1){10}}
\put(18,14){\color{red}\vector(0,1){10}}
\put(-2,14){\color{red}\line(1,0){8}}
\put(6,14){\color{red}\line(1,0){12}}
\put(18,14){\vector(1,0){11}}
\put(6,7){\line(0,1){7}}
\put(18,7){\line(0,1){7}}
\end{picture}}
\end{picture}
= 1 \st 1 + z {\bf a}^+\st 1+ zq{\bf k} \st {\bf a}^+, 
\end{equation*}
\begin{equation*}
\!\!\!\!\!\!\!\!\!\!\!\!\!\!
\!\!\!\!\!\!\!\!\!\!\!\!\!\!\!\!\!\!\!
S^0_{10}(z) =
\begin{picture}(150,20)(0,0)
\thicklines
\mbox{\begin{picture}(40,18)(-6,12)
\put(6,14){\color{red}\vector(0,1){10}}
\put(18,14){\vector(0,1){10}}
\put(-2,14){\line(1,0){8}}
\put(6,14){\line(1,0){12}}
\put(18,14){\vector(1,0){11}}
\put(6,7){\color{red}\line(0,1){7}}
\put(18,7){\line(0,1){7}}
\end{picture}}
+
\mbox{\begin{picture}(40,18)(-6,12)
\put(6,14){\vector(0,1){10}}
\put(18,14){\color{red}\vector(0,1){10}}
\put(-2,14){\line(1,0){8}}
\put(6,14){\color{red}\line(1,0){12}}
\put(18,14){\vector(1,0){11}}
\put(6,7){\color{red}\line(0,1){7}}
\put(18,7){\line(0,1){7}}
\end{picture}}
+
\mbox{\begin{picture}(40,18)(-6,12)
\put(6,14){\color{red}\vector(0,1){10}}
\put(18,14){\color{red}\vector(0,1){10}}
\put(-2,14){\color{red}\line(1,0){8}}
\put(6,14){\color{red}\line(1,0){12}}
\put(18,14){\vector(1,0){11}}
\put(6,7){\color{red}\line(0,1){7}}
\put(18,7){\line(0,1){7}}
\end{picture}}
\end{picture}
= {\bf k} \st 1 + {\bf a}^-\st {\bf a}^+ + z 1 \st {\bf a}^+, 
\end{equation*}
\begin{equation*}
\hspace{-4.8cm}
S^1_{10}(z) =
\begin{picture}(90,20)(00,0)
\thicklines
\mbox{\begin{picture}(40,18)(-6,12)
\put(6,14){\vector(0,1){10}}
\put(18,14){\vector(0,1){10}}
\put(-2,14){\line(1,0){8}}
\put(6,14){\color{red}\line(1,0){12}}
\put(18,14){\color{red}\vector(1,0){11}}
\put(6,7){\color{red}\line(0,1){7}}
\put(18,7){\line(0,1){7}}
\end{picture}}
+
\mbox{\begin{picture}(40,18)(-6,12)
\put(6,14){\color{red}\vector(0,1){10}}
\put(18,14){\vector(0,1){10}}
\put(-2,14){\color{red}\line(1,0){8}}
\put(6,14){\color{red}\line(1,0){12}}
\put(18,14){\color{red}\vector(1,0){11}}
\put(6,7){\color{red}\line(0,1){7}}
\put(18,7){\line(0,1){7}}
\end{picture}}
\end{picture}
\;= qz^{-1}{\bf a}^- \st {\bf k} 
+ q 1 \st {\bf k},
\end{equation*}
\begin{equation*}
\hspace{-8.7cm}
S^1_{00}(z) =
\begin{picture}(35,20)(0,0)
\thicklines
\mbox{\begin{picture}(40,18)(-6,12)
\put(6,14){\vector(0,1){10}}
\put(18,14){\vector(0,1){10}}
\put(-2,14){\color{red}\line(1,0){8}}
\put(6,14){\color{red}\line(1,0){12}}
\put(18,14){\color{red}\vector(1,0){11}}
\put(6,7){\line(0,1){7}}
\put(18,7){\line(0,1){7}}
\end{picture}}
\end{picture}
\;= q^2 {\bf k}\st {\bf k}.  
\end{equation*}
$ $
\end{example}

\begin{example}\label{ex:ts}
Consider the case $(m,n)=(2,2)$. 
$S(z)^{00}_{00}$ consists of the following 8 terms:
\begin{equation*}
\begin{picture}(500,50)(0,-20)
\thicklines
\mbox{\begin{picture}(37,18)(-5,5)
\put(6,14){\vector(0,1){10}}
\put(18,14){\vector(0,1){10}}
\put(-2,14){\line(1,0){8}}
\put(6,14){\line(1,0){12}}
\put(18,14){\vector(1,0){11}}
\put(6,0){\line(0,1){14}}
\put(18,0){\line(0,1){14}}
\put(-2,0){\line(1,0){8}}
\put(6,0){\line(1,0){12}}
\put(18,0){\vector(1,0){11}}
\put(6,-8){\line(0,1){8}}
\put(18,-8){\line(0,1){8}}
\end{picture}}
+
\mbox{\begin{picture}(37,18)(-5,5)
\put(6,14){\color{red}\vector(0,1){10}}
\put(18,14){\vector(0,1){10}}
\put(-2,14){\color{red}\line(1,0){8}}
\put(6,14){\line(1,0){12}}
\put(18,14){\vector(1,0){11}}
\put(6,0){\line(0,1){14}}
\put(18,0){\line(0,1){14}}
\put(-2,0){\line(1,0){8}}
\put(6,0){\line(1,0){12}}
\put(18,0){\vector(1,0){11}}
\put(6,-8){\line(0,1){8}}
\put(18,-8){\line(0,1){8}}
\end{picture}}
+
\mbox{\begin{picture}(37,18)(-5,5)
\put(6,14){\color{red}\vector(0,1){11}}
\put(18,14){\vector(0,1){11}}
\put(-2,14){\line(1,0){8}}
\put(6,14){\line(1,0){12}}
\put(18,14){\vector(1,0){11}}
\put(6,0){\color{red}\line(0,1){14}}
\put(18,0){\line(0,1){14}}
\put(-2,0){\color{red}\line(1,0){8}}
\put(6,0){\line(1,0){12}}
\put(18,0){\vector(1,0){11}}
\put(6,-8){\line(0,1){8}}
\put(18,-8){\line(0,1){8}}
\end{picture}}
+
\mbox{\begin{picture}(37,18)(-5,5)
\put(6,14){\vector(0,1){10}}
\put(18,14){\color{red}\vector(0,1){10}}
\put(-2,14){\line(1,0){8}}
\put(6,14){\color{red}\line(1,0){12}}
\put(18,14){\vector(1,0){11}}
\put(6,0){\color{red}\line(0,1){14}}
\put(18,0){\line(0,1){14}}
\put(-2,0){\color{red}\line(1,0){8}}
\put(6,0){\line(1,0){12}}
\put(18,0){\vector(1,0){11}}
\put(6,-8){\line(0,1){8}}
\put(18,-8){\line(0,1){8}}
\end{picture}}
+
\mbox{\begin{picture}(37,18)(-5,5)
\put(6,14){\color{red}\vector(0,1){10}}
\put(18,14){\color{red}\vector(0,1){10}}
\put(-2,14){\color{red}\line(1,0){8}}
\put(6,14){\color{red}\line(1,0){12}}
\put(18,14){\vector(1,0){11}}
\put(6,0){\color{red}\line(0,1){14}}
\put(18,0){\line(0,1){14}}
\put(-2,0){\color{red}\line(1,0){8}}
\put(6,0){\line(1,0){12}}
\put(18,0){\vector(1,0){11}}
\put(6,-8){\line(0,1){8}}
\put(18,-8){\line(0,1){8}}
\end{picture}}
+
\mbox{\begin{picture}(37,18)(-5,5)
\put(6,14){\vector(0,1){10}}
\put(18,14){\color{red}\vector(0,1){10}}
\put(-2,14){\line(1,0){8}}
\put(6,14){\line(1,0){12}}
\put(18,14){\vector(1,0){11}}
\put(6,0){\line(0,1){14}}
\put(18,0){\color{red}\line(0,1){14}}
\put(-2,0){\color{red}\line(1,0){8}}
\put(6,0){\color{red}\line(1,0){12}}
\put(18,0){\vector(1,0){11}}
\put(6,-8){\line(0,1){8}}
\put(18,-8){\line(0,1){8}}
\end{picture}}
+
\mbox{\begin{picture}(37,18)(-5,5)
\put(6,14){\color{red}\vector(0,1){10}}
\put(18,14){\color{red}\vector(0,1){10}}
\put(-2,14){\color{red}\line(1,0){8}}
\put(6,14){\line(1,0){12}}
\put(18,14){\vector(1,0){11}}
\put(6,0){\line(0,1){14}}
\put(18,0){\color{red}\line(0,1){14}}
\put(-2,0){\color{red}\line(1,0){8}}
\put(6,0){\color{red}\line(1,0){12}}
\put(18,0){\vector(1,0){11}}
\put(6,-8){\line(0,1){8}}
\put(18,-8){\line(0,1){8}}
\end{picture}}
+
\mbox{\begin{picture}(37,18)(-5,5)
\put(6,14){\vector(0,1){10}}
\put(18,14){\color{red}\vector(0,1){10}}
\put(-2,14){\color{red}\line(1,0){8}}
\put(6,14){\color{red}\line(1,0){12}}
\put(18,14){\vector(1,0){11}}
\put(6,0){\line(0,1){14}}
\put(18,0){\line(0,1){14}}
\put(-2,0){\line(1,0){8}}
\put(6,0){\line(1,0){12}}
\put(18,0){\vector(1,0){11}}
\put(6,-8){\line(0,1){8}}
\put(18,-8){\line(0,1){8}}
\end{picture}}
\end{picture}
\end{equation*}
Thus we have
\begin{equation*}
\begin{split}
S(z)^{00}_{00} &= 
1\st 1 \st 1 \st1 
+ z {\bf a}^+ \st 1 \st 1 \st 1 
+ z{\bf k} \st {\bf a}^+ \st 1 \st 1 
+ z{\bf a}^- \st {\bf a}^+\st {\bf a}^+ \st 1\\
&+z^2 1 \st {\bf a}^+ \st {\bf a}^+ \st 1
+qz 1 \st {\bf k}\st {\bf k} \st {\bf a}^+
+qz^2 {\bf a}^+ \st {\bf k} \st {\bf k} \st {\bf a}^+
+qz {\bf k}\st 1 \st {\bf a}^+ \st 1.
\end{split}
\end{equation*}
\end{example}

\begin{example}\label{ex:ts2}
Similarly $S(z)^{10}_{10}$ for $(m,n)=(2,2)$ consists of the following 8 terms:
\begin{equation*}
\begin{picture}(500,50)(0,-20)
\thicklines
\mbox{\begin{picture}(37,18)(-5,5)
\put(6,14){\vector(0,1){10}}
\put(18,14){\vector(0,1){10}}
\put(-2,14){\line(1,0){8}}
\put(6,14){\line(1,0){12}}
\put(18,14){\color{red}\vector(1,0){11}}
\put(6,0){\line(0,1){14}}
\put(18,0){\color{red}\line(0,1){14}}
\put(-2,0){\line(1,0){8}}
\put(6,0){\color{red}\line(1,0){12}}
\put(18,0){\vector(1,0){11}}
\put(6,-8){\color{red}\line(0,1){8}}
\put(18,-8){\line(0,1){8}}
\end{picture}}
+
\mbox{\begin{picture}(37,18)(-5,5)
\put(6,14){\color{red}\vector(0,1){10}}
\put(18,14){\vector(0,1){10}}
\put(-2,14){\color{red}\line(1,0){8}}
\put(6,14){\line(1,0){12}}
\put(18,14){\color{red}\vector(1,0){11}}
\put(6,0){\line(0,1){14}}
\put(18,0){\color{red}\line(0,1){14}}
\put(-2,0){\line(1,0){8}}
\put(6,0){\color{red}\line(1,0){12}}
\put(18,0){\vector(1,0){11}}
\put(6,-8){\color{red}\line(0,1){8}}
\put(18,-8){\line(0,1){8}}
\end{picture}}
+
\mbox{\begin{picture}(37,18)(-5,5)
\put(6,14){\color{red}\vector(0,1){10}}
\put(18,14){\vector(0,1){10}}
\put(-2,14){\line(1,0){8}}
\put(6,14){\line(1,0){12}}
\put(18,14){\color{red}\vector(1,0){11}}
\put(6,0){\color{red}\line(0,1){14}}
\put(18,0){\color{red}\line(0,1){14}}
\put(-2,0){\color{red}\line(1,0){8}}
\put(6,0){\color{red}\line(1,0){12}}
\put(18,0){\vector(1,0){11}}
\put(6,-8){\color{red}\line(0,1){8}}
\put(18,-8){\line(0,1){8}}
\end{picture}}
+
\mbox{\begin{picture}(37,18)(-5,5)
\put(6,14){\vector(0,1){10}}
\put(18,14){\color{red}\vector(0,1){10}}
\put(-2,14){\line(1,0){8}}
\put(6,14){\color{red}\line(1,0){12}}
\put(18,14){\color{red}\vector(1,0){11}}
\put(6,0){\color{red}\line(0,1){14}}
\put(18,0){\color{red}\line(0,1){14}}
\put(-2,0){\color{red}\line(1,0){8}}
\put(6,0){\color{red}\line(1,0){12}}
\put(18,0){\vector(1,0){11}}
\put(6,-8){\color{red}\line(0,1){8}}
\put(18,-8){\line(0,1){8}}
\end{picture}}
+
\mbox{\begin{picture}(37,18)(-5,5)
\put(6,14){\color{red}\vector(0,1){10}}
\put(18,14){\color{red}\vector(0,1){10}}
\put(-2,14){\color{red}\line(1,0){8}}
\put(6,14){\color{red}\line(1,0){12}}
\put(18,14){\color{red}\vector(1,0){11}}
\put(6,0){\color{red}\line(0,1){14}}
\put(18,0){\color{red}\line(0,1){14}}
\put(-2,0){\color{red}\line(1,0){8}}
\put(6,0){\color{red}\line(1,0){12}}
\put(18,0){\vector(1,0){11}}
\put(6,-8){\color{red}\line(0,1){8}}
\put(18,-8){\line(0,1){8}}
\end{picture}}
+
\mbox{\begin{picture}(37,18)(-5,5)
\put(6,14){\color{red}\vector(0,1){10}}
\put(18,14){\vector(0,1){10}}
\put(-2,14){\color{red}\line(1,0){8}}
\put(6,14){\color{red}\line(1,0){12}}
\put(18,14){\color{red}\vector(1,0){11}}
\put(6,0){\color{red}\line(0,1){14}}
\put(18,0){\line(0,1){14}}
\put(-2,0){\line(1,0){8}}
\put(6,0){\line(1,0){12}}
\put(18,0){\vector(1,0){11}}
\put(6,-8){\color{red}\line(0,1){8}}
\put(18,-8){\line(0,1){8}}
\end{picture}}
+
\mbox{\begin{picture}(37,18)(-5,5)
\put(6,14){\vector(0,1){10}}
\put(18,14){\color{red}\vector(0,1){10}}
\put(-2,14){\color{red}\line(1,0){8}}
\put(6,14){\color{red}\line(1,0){12}}
\put(18,14){\color{red}\vector(1,0){11}}
\put(6,0){\line(0,1){14}}
\put(18,0){\color{red}\line(0,1){14}}
\put(-2,0){\line(1,0){8}}
\put(6,0){\color{red}\line(1,0){12}}
\put(18,0){\vector(1,0){11}}
\put(6,-8){\color{red}\line(0,1){8}}
\put(18,-8){\line(0,1){8}}
\end{picture}}
+
\mbox{\begin{picture}(37,18)(-5,5)
\put(6,14){\vector(0,1){10}}
\put(18,14){\vector(0,1){10}}
\put(-2,14){\line(1,0){8}}
\put(6,14){\color{red}\line(1,0){12}}
\put(18,14){\color{red}\vector(1,0){11}}
\put(6,0){\color{red}\line(0,1){14}}
\put(18,0){\line(0,1){14}}
\put(-2,0){\line(1,0){8}}
\put(6,0){\line(1,0){12}}
\put(18,0){\vector(1,0){11}}
\put(6,-8){\color{red}\line(0,1){8}}
\put(18,-8){\line(0,1){8}}
\end{picture}}
\end{picture}
\end{equation*}
Thus we have
\begin{equation*}
\begin{split}
S(z)^{10}_{10} &= 
z^{-1}1\st {\bf a}^-  \st {\bf a}^-  \st{\bf a}^+
+ {\bf a}^+ \st {\bf a}^-  \st {\bf a}^-  \st {\bf a}^+
+ {\bf k} \st 1 \st {\bf a}^-  \st {\bf a}^+  
+ {\bf a}^- \st 1 \st 1 \st {\bf a}^+ \\
&+z 1 \st 1 \st 1 \st {\bf a}^+
+q 1 \st {\bf k}\st {\bf k} \st 1
+q {\bf k} \st {\bf a}^- \st 1 \st {\bf a}^+
+qz^{-1}{\bf a}^- \st {\bf k} \st {\bf k} \st 1.
\end{split}
\end{equation*}
\end{example}

The layer to layer transfer matrices $S(z)^{\bf a}_{\bf j}$ 
with the common SE boundary condition ${\bf a}, {\bf j}$ 
form a commuting family. 
\begin{proposition}[Commutativity of layer to layer transfer matrices] 
\label{prop:S comm}
\begin{equation} \label{S comm}
[S(x)^{\bf a}_{\bf j},S(y)^{\bf a}_{\bf j}]=0.
\end{equation}
\end{proposition}
\begin{proof}
This is a consequence of the tetrahedron equation in 
theorem \ref{prop:tetrahedron} and the `trivial' eigenvectors of 
$\M(z)$ in propositions \ref{prop:right eigen} and \ref{prop:left eigen}.
Consider the following two operators on $F^{\ot mn}\ot F$.
\begin{align}
\sum_{\bf b,b'}
&\Bigl(\M({\textstyle\frac{x}{x'}})_{a_m,a_m}^{a_m,a_m}\cdots
\M({\textstyle\frac{x}{x'}})_{a_1,a_1}^{a_1,a_1}\Bigr)
\Bigl(\M({\textstyle\frac{y}{y'}})_{b_n,b'_n}^{c_n,c'_n}\cdots
\M({\textstyle\frac{y}{y'}})_{b_1,b'_1}^{c_1,c'_1}\Bigr)
T({\textstyle\frac{x}{y}})_{\bf i,j}^{\bf a,b}
T({\textstyle\frac{x'}{y'}})_{\bf i',j}^{\bf a,b'},
\label{lhs}\\
\sum_{\bf k,k'}
&T({\textstyle\frac{x'}{y'}})_{\bf k',j}^{\bf a,c'}
T({\textstyle\frac{x}{y}})_{\bf k,j}^{\bf a,c}
\Bigl(\M({\textstyle\frac{y}{y'}})_{j_n,j_n}^{j_n,j_n}\cdots
\M({\textstyle\frac{y}{y'}})_{j_1,j_1}^{j_1,j_1}\Bigr)
\Bigl(\M({\textstyle\frac{x}{x'}})_{i_m,i'_m}^{k_m,k'_m}\cdots
\M({\textstyle\frac{x}{x'}})_{i_1,i'_1}^{k_1,k'_1}\Bigr),
\label{rhs}
\end{align}
where ${\bf i}=(i_1,\ldots,i_m)$, etc. 
The left block of $\M$ both in \eqref{lhs} and
\eqref{rhs} are actually the identities 
but it is better to keep them temporarily for the proof.
The operators in (\ref{lhs}) and (\ref{rhs}) coincide.
To see this we depict them as follows.
\begin{equation} \label{fig}
\begin{picture}(400,135)(-43,-10)
\setlength{\unitlength}{0.6mm}
\thicklines
{\color{green}\drawline(0,60)(60,60)
\qbezier(60,60)(63,60)(63,57)\put(63,57){\vector(0,-1){60}}}

\put(-21,32){$\displaystyle{\sum_{{\bf b},{\bf b}'}}$}
\put(-10,25){
\qbezier(67,18)(68,18)(69,18)\put(69,18){\vector(1,-1){10}}
\put(14,18){\line(1,0){5}}\put(21,18){\line(1,0){26}}\put(49,18){\line(1,0){18}}
\put(8,17){$i'_1$}\put(64,20){$a_1$}\put(84.5,16){$a_1$}
\put(10,10){\line(1,0){50}}
\qbezier(60,10)(62,10.2)(63,10.3)
\put(63,10.3){\vector(3,1){20}}
\put(4,8.2){$i_1$}\put(63,6){$a_1$}\put(80,5){$a_1$}
\put(12,-3){$\vdots$}
}

\put(-10,-3){
\qbezier(67,18)(68,18)(69,18)\put(69,18){\vector(1,-1){10}}
\put(14,18){\line(1,0){5}}\put(21,18){\line(1,0){26}}\put(49,18){\line(1,0){18}}
\put(6,17){$i'_m$}\put(64,20){$a_m$}\put(84.5,16){$a_m$}
\put(10,10){\line(1,0){50}}
\qbezier(60,10)(62,10.2)(63,10.3)
\put(63,10.3){\vector(3,1){20}}
\put(2.5,8){$i_m$}\put(63,6){$a_m$}\put(80,5){$a_m$}
}

\put(-10,-7){
\put(30.5,63){\vector(-1,1){9.5}}\qbezier(30.5,63)(30.9,62.7)(31,62)
\put(31,10){\line(0,1){3}}\put(31,15){\line(0,1){26}}\put(31,43){\line(0,1){19}}
\put(15,74){$c'_1$}\put(30,78){$c_1$}
\put(33,61){$b'_1$}\put(30,4){$j_1$}
\put(20.4,57){\vector(2,3){12}}\qbezier(20.4,57)(20.3,56.8)(20,56)
\put(20,6){\line(0,1){50}}
\put(14,57){$b_1$}\put(18,0){$j_1$}
\put(0.2,7){\rotatebox{36}{\put(41,25){\color{blue}\vector(-1,0){22}}}}
\put(0.2,-21){\rotatebox{36}{\put(41,25){\color{blue}\vector(-1,0){22}}}}
\put(35,53){$\cdots$}
}

\put(18,-7){
\put(30.5,63){\vector(-1,1){9.5}}\qbezier(30.5,63)(30.9,62.7)(31,62)
\put(31,10){\line(0,1){3}}\put(31,15){\line(0,1){26}}\put(31,43){\line(0,1){19}}
\put(15,74){$c'_n$}\put(30,78){$c_n$}
\put(33,61){$b'_n$}\put(30,4){$j_n$}
\put(20.4,57){\vector(2,3){12}}\qbezier(20.4,57)(20.3,56.8)(20,56)
\put(20,6){\line(0,1){50}}
\put(14,57){$b_n$}\put(18,0){$j_n$}
\put(0.2,7){\rotatebox{36}{\put(41,25){\color{blue}\vector(-1,0){22}}}}
\put(0.2,-21){\rotatebox{36}{\put(41,25){\color{blue}\vector(-1,0){22}}}}
}


\put(89,32){$=\;\;\;
\displaystyle{\sum_{{\bf k},{\bf k}'}}$}

\put(129,5){
{\color{green}\drawline(0,60)(0,3)
\qbezier(0,3)(0,0)(3,0)\put(3,0){\vector(1,0){60}}}

\put(-1,32){
\drawline(-8,12)(10,18)\qbezier(10,18)(11,18)(12.6,18)
\put(15,18){\line(1,0){26}}\put(43,18){\vector(1,0){24}}
\put(4,20){$k_1$}\put(69,17){$a_1$}
\drawline(-6,20)(6,12)\qbezier(6,12)(8,10)(10,10)
\put(10,10){\vector(1,0){50}}
\put(2.2,6){$k'_1$}\put(62,9){$a_1$}
\put(-11.5,21){$i'_1$}\put(-13.5,10){$i_1$}
\put(-5,0){$\vdots$}
}

\put(-1,4){
\drawline(-8,12)(10,18)\qbezier(10,18)(11,18)(12.6,18)
\put(15,18){\line(1,0){26}}\put(43,18){\vector(1,0){24}}
\put(4,20){$k_m$}\put(69,17){$a_m$}
\drawline(-6,20)(6,12)\qbezier(6,12)(8,10)(10,10)
\put(10,10){\vector(1,0){50}}
\put(1.8,6.2){$k'_m$}\put(62,9){$a_m$}
\put(-13.5,20){$i'_m$}\put(-15,10){$i_m$}
}

\put(-7,0){
\drawline(21,-10)(31,8)  \qbezier(31,8)(31,9) (31,10)
\put(31,10){\line(0,1){3}}\put(31,15){\line(0,1){26}}\put(31,43){\vector(0,1){19}}\put(29,64){$c_1$}\put(32.2,7){$j_1$}\put(32.3,-7.5){$j_1$}
\drawline(32,-5)(21,4)  \qbezier(21,4)(21,4) (20,6)
\put(20,6){\vector(0,1){50}}
\put(18,58){$c'_1$}\put(15.1,4){$j_1$}\put(16.3,-12){$j_1$}
\put(0.2,7){\rotatebox{36}{\put(41,25){\color{blue}\vector(-1,0){22}}}}
\put(0.2,-21){\rotatebox{36}{\put(41,25){\color{blue}\vector(-1,0){22}}}}
\put(36,55){$\cdots$}
}

\put(21,0){
\drawline(21,-10)(31,8)  \qbezier(31,8)(31,9) (31,10)
\put(31,10){\line(0,1){3}}\put(31,15){\line(0,1){26}}\put(31,43){\vector(0,1){19}}
\put(29,64){$c_n$}\put(32.2,7){$j_n$}\put(32.3,-7.5){$j_n$}
\drawline(32,-5)(21,4)  \qbezier(21,4)(21,4) (20,6)
\put(20,6){\vector(0,1){50}}
\put(18,58){$c'_n$}\put(14.2,3.6){$j_n$}\put(16.1,-12){$j_n$}
\put(0.2,7){\rotatebox{36}{\put(41,25){\color{blue}\vector(-1,0){22}}}}
\put(0.2,-21){\rotatebox{36}{\put(41,25){\color{blue}\vector(-1,0){22}}}}
}

}
\end{picture}
\thinlines
\end{equation}
Here $T(z)_{\bf i,j}^{\bf a,b}$ acts on $F^{\ot mn}$ (blue arrows) and
$\M(z)_{i,j}^{a,b}$ acts on the extra single Fock space $F$ (green arrow).
In the left figure,
the front and the back layers correspond to 
$T({\textstyle\frac{x}{y}})_{\bf i,j}^{\bf a,b}$ and 
$T({\textstyle\frac{x'}{y'}})_{\bf i',j}^{\bf a,b'}$ in (\ref{lhs}), respectively.
Similarly in the right figure,
the front and the back layers represent  
$T({\textstyle\frac{x'}{y'}})_{\bf k',j}^{\bf a,c'}$ and 
$T({\textstyle\frac{x}{y}})_{\bf k,j}^{\bf a,c}$ in (\ref{rhs}), respectively.
From the top right corner of the left figure, using theorem \ref{prop:tetrahedron} 
one can move the green line all the way down to the bottom left. It means that
the left figure and the right one are equal as operators on $F^{\ot mn}\ot F$.
Now we rephrase \eqref{lhs}=\eqref{rhs} as
\begin{equation}\label{syk}
\begin{split}
&\sum_{\bf b,b'}
\Bigl(\M({\textstyle\frac{y}{y'}})_{b_n,b'_n}^{c_n,c'_n}\cdots
\M({\textstyle\frac{y}{y'}})_{b_1,b'_1}^{c_1,c'_1}\Bigr)
T({\textstyle\frac{x}{y}})_{\bf i,j}^{\bf a,b}
T({\textstyle\frac{x'}{y'}})_{\bf i',j}^{\bf a,b'}\\
=
&\sum_{\bf k,k'}
T({\textstyle\frac{x'}{y'}})_{\bf k',j}^{\bf a,c'}
T({\textstyle\frac{x}{y}})_{\bf k,j}^{\bf a,c}
\Bigl(\M({\textstyle\frac{x}{x'}})_{i_m,i'_m}^{k_m,k'_m}\cdots
\M({\textstyle\frac{x}{x'}})_{i_1,i'_1}^{k_1,k'_1}\Bigr)
\end{split}
\end{equation}
removing the identity parts.
Evaluate (\ref{syk}) between 
$\langle\chi({\textstyle\frac{y}{y'}})|\in F^*$ and 
$|\chi({\textstyle\frac{x}{x'}})\rangle \in F$,
where these vectors are on the green arrows on which only the block of  
$\M(z)$'s act.
Further taking the sum over ${\bf i,i',c,c'}$ on the both sides 
by means of (\ref{mst}) we find
\[
\langle\chi({\textstyle\frac{y}{y'}})|\chi({\textstyle\frac{x}{x'}})\rangle
\sum_{\bf i,i',b,b'}T({\textstyle \frac{x}{y}})_{\bf i,j}^{\bf a,b}
T({\textstyle\frac{x'}{y'}})_{\bf i',j}^{\bf a,b'}
=
\langle\chi({\textstyle\frac{y}{y'}})|\chi({\textstyle\frac{x}{x'}})\rangle
\sum_{\bf k,k',c,c'}T({\textstyle\frac{x'}{y'}})_{\bf k',j}^{\bf a,c'}
T({\textstyle\frac{x}{y}})_{\bf k,j}^{\bf a,c}.
\]
Since 
$\langle\chi({\textstyle\frac{y}{y'}})|\chi({\textstyle\frac{x}{x'}})\rangle
= \sum_{m\ge 0}\frac{(q^2)_m}{(q)_m^2}({\textstyle\frac{xy}{x'y'}})^m \neq 0$,
we get 
$S({\textstyle \frac{x}{y}})^{\bf a}_{\bf j}
S({\textstyle\frac{x'}{y'}})^{\bf a}_{\bf j}
=S({\textstyle\frac{x'}{y'}})^{\bf a}_{\bf j}
S({\textstyle \frac{x}{y}})^{\bf a}_{\bf j}$ by (\ref{S}).
\end{proof}

\begin{example}\label{ex:ymi}
The commutativity (\ref{S comm}) is easily seen for 
those $S(z)^{\bf a}_{\bf j}$ in examples \ref{ex:ts0} and \ref{ex:nmi}.
Let us check it for 
$S^{00}_{00}(z) = \sum_{i=0}^2 z^i W_i$ in example \ref{ex:ts}.
The relation $[W_i,W_j]=0$ to be shown is nontrivial only for $(i,j)=(1,2)$.
We have 
$W_2 = 1 \ot {\bf a}^+ \ot {\bf a}^+ \ot 1 
+ q {\bf a}^+ \ot {\bf k} \ot {\bf k} \ot {\bf a}^+$ and  
$W_1 = {\bf a}^+\ot 1 \ot 1\ot 1+ U \ot 1 + q V \ot {\bf a}^+$, 
where 
$U = {\bf k}\ot {\bf a}^+ \ot 1 + {\bf a}^- \ot {\bf a}^+ \ot {\bf a}^+ 
+ q {\bf k}\ot 1 \ot {\bf a}^+$ and $V = 1 \ot {\bf k} \ot {\bf k}$.
As $[ {\bf a}^+ \ot 1\ot 1\ot 1, W_2]=0$, we are to show
\begin{align*}
0&=[U \ot 1 + q V \ot {\bf a}^+, 
1 \ot {\bf a}^+ \ot {\bf a}^+ \ot 1 
+ q {\bf a}^+ \ot {\bf k} \ot {\bf k} \ot {\bf a}^+]\\
&= [U, 1 \ot {\bf a}^+ \ot {\bf a}^+]\ot 1 + 
q Y \ot {\bf a}^+ + q^2 [V, {\bf a}^+ \ot {\bf k} \ot {\bf k} ]\ot ({\bf a}^+)^2,
\end{align*}
where $Y = [U, {\bf a}^+ \ot {\bf k} \ot {\bf k}] + 
[V, 1 \ot {\bf a}^+ \ot {\bf a}^+]$.
Obviously the leftmost and the rightmost commutators in the last expression
vanish.
Hence we are to show $Y=0$.
The relation ${\bf k} \,{\bf a}^+ = -q {\bf a}^+ {\bf k}$ in (\ref{tgm})
tells $[{\bf k}\ot {\bf a}^+ \ot 1,{\bf a}^+ \ot {\bf k} \ot {\bf k}] 
= [{\bf k}\ot 1 \ot {\bf a}^+,{\bf a}^+ \ot {\bf k} \ot {\bf k}] =0$.
Thus $Y=0$ reduces to 
$[{\bf a}^- \ot {\bf a}^+ \ot {\bf a}^+,{\bf a}^+ \ot {\bf k} \ot {\bf k}]
+[V, 1 \ot {\bf a}^+ \ot {\bf a}^+]=0$.
This is straightforward by (\ref{tgm}).
\end{example}

\section{Further bilinear relations}\label{sec:fb}
In the proof of proposition \ref{prop:S comm}
we used the fact that $v_i\ot v_i\ot |\xi\rangle$ and
$v^*_i\ot v^*_i\ot \langle\xi|$ are eigenvectors of $\M(z)$. 
However, there are also the third eigenvectors in both
proposition \ref{prop:right eigen} and \ref{prop:left eigen}
which are slightly more involved.
By using them we can generate further bilinear relations among 
the layer to layer transfer matrices 
$S(z)^{\bf a}_{\bf j}$'s mixing different boundary conditions ${\bf a},{\bf j}$. 
To describe such relations we prepare some notation.
 
Recall that $m$ and $n$ are any positive integers 
representing the size of the layer as in (\ref{S}).
For a subset $I \subseteq \{1,\ldots,m\}$ with the 
complement $\overline{I}=\{1,\ldots, m\} \setminus I$ and 
sequences $\grbold{\alpha} \in \{0,1\}^{\# I}$, 
$\grbold{\beta} \in \{0,1\}^{\# \overline{I}}$,
let $(\grbold{\alpha}_I,\grbold{\beta}_{\overline{I}}) 
\in \{0,1\}^m$ 
be the sequence in which the subsequence corresponding to the indices in 
$I$ is $\grbold{\alpha}$ and the rest $\overline{I}$ is $\grbold{\beta}$. 
For instance for $m=5$, 
$(\grbold{\alpha}_{\{1,3,4\}},\grbold{\beta}_{\{2,5\}})
=(\alpha_1,\beta_1,\alpha_2,\alpha_3,\beta_2)$
for $I=\{1,3,4\}$, 
$\grbold{\alpha}=(\alpha_1,\alpha_2,\alpha_3)$ and 
$\grbold{\beta}=(\beta_1,\beta_2)$\footnote{
Note that it is {\em not} $(\alpha_1,\alpha_3,\alpha_4,\beta_2,\beta_5)$.}.
Likewise for $J \sqcup \overline{J} = \{1,\ldots, n\}$
and $\grbold{\gamma} \in \{0,1\}^{\# J}$, 
$\grbold{\delta} \in \{0,1\}^{\# \overline{J}}$, 
the symbol 
$(\grbold{\gamma}_J,\grbold{\delta}_{\overline{J}}) \in \{0,1\}^n$ 
denotes the similar sequence.
For any sequence $\grbold{\alpha}
=(\alpha_1,\ldots,\alpha_k)\in \{0,1\}^k$, we set 
$|\grbold{\alpha}|=\alpha_1+\cdots+\alpha_k$ and
$\overline{\grbold{\alpha}}=(1-\alpha_1,\ldots,1-\alpha_k)$.

\begin{theorem} [Bilinear relations of layer to layer transfer matrices]\label{prop:S rel}
For any subsets $I \subseteq \{1,\ldots,m\}$ and 
$J \subseteq \{1,\ldots,n\}$ and sequences 
$\grbold{\alpha}\in \{0,1\}^{\# I}$ and 
$\grbold{\gamma} \in \{0,1\}^{\# J}$, we have 
(parentheses omitted in suffixes of $S$)
\begin{equation} \label{S rel}
\sum_{\scriptsize\grbold{\beta,\delta}}
y^{|{\scriptsize\grbold{\beta}}|+|{\scriptsize\grbold{\delta}}|}
x^{|\overline{\scriptsize\grbold{\beta}}|+|\overline{\scriptsize\grbold{\delta}}|}\,
S(y)^{\scriptsize \grbold{\alpha}_I,\grbold{\beta}_{\overline{I}}}_{
\scriptsize \grbold{\gamma}_J,\grbold{\delta}_{\overline{J}}}\,
S(x)^{\scriptsize \grbold{\alpha}_I,\grbold{\overline{\beta}}_{\overline{I}}}_{
\scriptsize \grbold{\gamma}_J,\grbold{\overline{\delta}}_{\overline{J}}}
=(x \longleftrightarrow y),
\end{equation}
where the sum runs over 
$\grbold{\beta} \in \{0,1\}^{\# \overline{I}}$ and 
$\grbold{\delta} \in \{0,1\}^{\# \overline{J}}$.
\end{theorem}

Proposition \ref{prop:S comm} is the simplest case of theorem \ref{prop:S rel}
corresponding to $I=\{1,\ldots,m\},J=\{1,\ldots,n\}$, where the sum reduces 
to a single term.
As another example,
when $(m,n)=(4,3), I=\{1,3\},J=\{2,3\},\grbold{\alpha}=(0,1),\grbold{\gamma}=(1,0)$,
the relation \eqref{S rel} reads 
\begin{align*}
&x^3S(y)^{0010}_{010}S(x)^{0111}_{110}
+yx^2S(y)^{0011}_{010}S(x)^{0110}_{110}
+yx^2S(y)^{0110}_{010}S(x)^{0011}_{110}\\
&+y^2xS(y)^{0111}_{010}S(x)^{0010}_{110}
+yx^2S(y)^{0010}_{110}S(x)^{0111}_{010}
+y^2xS(y)^{0011}_{110}S(x)^{0110}_{010}\\
&+y^2xS(y)^{0110}_{110}S(x)^{0011}_{010}
+y^3S(y)^{0111}_{110}S(x)^{0010}_{010}
=(x \longleftrightarrow y).
\end{align*}

We will present a proof of theorem \ref{prop:S rel} 
only for the special case considered in 
corollary \ref{cor:S rel} below, 
since the general case is easily inferred from it.
It corresponds to the choice 
$I=\{2,3,\ldots, m\}$, $\grbold{\alpha} = {\bf a}$, 
$J = \{2,3,\ldots, n\}$, $\grbold{\gamma} = {\bf j}$ in (\ref{S rel}), 
which suffices for our application to TASEP in the next section.

\begin{corollary} \label{cor:S rel}
For any sequences ${\bf a} \in \{0,1\}^{m-1}$ and 
${\bf j} \in \{0,1\}^{n-1}$, we have
\begin{align*}
&x^2S(y)^{0\,{\bf a}}_{0\,{\bf j}}S(x)^{1\,{\bf a}}_{1\,{\bf j}} 
+yxS(y)^{0\,{\bf a}}_{1\,{\bf j}} S(x)^{1\,{\bf a}}_{0\,{\bf j}} \\ 
+\, &yxS(y)^{1\,{\bf a}}_{0\,{\bf j}}S(x)^{0\,{\bf a}}_{1\,{\bf j}} 
+y^2S(y)^{1\,{\bf a}}_{1\,{\bf j}}S(x)^{0\,{\bf a}} _{0\,{\bf j}} 
=(x \longleftrightarrow y).
\end{align*}
\end{corollary}
\begin{proof}
The proof proceeds similarly to that of proposition \ref{prop:S comm}. 
Consider the following equality of operators on
$F^{\ot mn}\ot F$.
\begin{equation}\label{MTT=TTM}
\begin{split}
&\sum_{{\bf b,b'}\atop a''_1+a'''_1=1}
\M({\textstyle \frac{x}{x'}})_{a''_1,a'''_1}^{a_1,a'_1}
\Bigl(\M({\textstyle \frac{y}{y'}})_{b_n,b'_n}^{c_n,c'_n}\cdots
\M({\textstyle \frac{y}{y'}})_{b_1,b'_1}^{c_1,c'_1}\Bigr)
T({\textstyle \frac{x}{y}})_{\bf i,j}^{\bf a'',b}
T({\textstyle \frac{x'}{y'}})_{\bf i',j'}^{\bf a''',b'}\\
=&\sum_{{\bf k,k'}\atop j''_1+j'''_1=1}
T({\textstyle \frac{x'}{y'}})_{\bf k',j'''}^{\bf a',c'}
T({\textstyle \frac{x}{y}})_{\bf k,j''}^{\bf a,c}
\M({\textstyle \frac{y}{y'}})_{j_1,j'_1}^{j''_1,j'''_1}
\Bigl(\M({\textstyle \frac{x}{x'}})_{i_m,i'_m}^{k_m,k'_m}\cdots
\M({\textstyle \frac{x}{x'}})_{i_1,i'_1}^{k_1,k'_1}\Bigr),
\end{split}
\end{equation}
where ${\bf a,a',a'',a'''}$ (resp. ${\bf j,j',j'',j'''}$)\footnote{${\bf a ,j}$ here 
have a different meaning from those in the statement.}
differ from each other only at the first component, which are
$a_1,a'_1,a''_1,a'''_1$ (resp. $j_1,j'_1,j''_1,j'''_1$). 
We take $a_1+a'_1=1$ and $j_1+j'_1=1$ and exhibited the
constraints coming from (\ref{ice}).
Unlike the previous 
$\eqref{lhs}=\eqref{rhs}$, the identity operators were not written. 
The difference of (\ref{MTT=TTM}) from 
\eqref{fig} is that the pair $(a_1,a_1)$ on the right end 
and $(j_1,j_1)$ at the bottom left 
were changed to $(a_1,a'_1)$ and $(j_1,j'_1)$
which should necessarily be $(0,1)$ or $(1,0)$.

Substitution of $\alpha=xy',\beta=x'y$ into 
(\ref{sin1}) and (\ref{sin2}) lead to
\begin{align}\label{szk}
\sum_{i+j=1}\alpha^i\beta^j
\M({\textstyle \frac{y}{y'}})_{i,j}^{k,l}|
\chi({\textstyle \frac{x}{x'}})\rangle=\alpha^k\beta^l
|\chi({\textstyle \frac{x}{x'}})\rangle,
\quad\sum_{i+j=1}\alpha^i\beta^j\langle\chi({\textstyle \frac{y'}{y}})|\,
\M({\textstyle \frac{x}{x'}})^{i,j}_{k,l}=\alpha^k\beta^l
\langle\chi({\textstyle \frac{y'}{y}})|.
\end{align}
On the both sides of \eqref{MTT=TTM},  
multiply $\alpha^{a_1+j_1}\beta^{a'_1+j'_1}$ 
and take sum over ${\bf i, i', c, c'}$ and 
$a_1,a'_1,j_1,j'_1$ with the constraints $a_1+a'_1=1,j_1+j'_1=1$. 
Evaluate the matrix element of the resulting operator identity 
between $\langle\chi({\textstyle \frac{y'}{y}})|$ from the left and
$|\chi({\textstyle \frac{x}{x'}})\rangle$ from the right.
Thanks to the identities (\ref{mst}) and (\ref{szk}),
all the $M$ operators disappear.
After canceling 
$\langle\chi({\textstyle \frac{y'}{y}})|
\chi({\textstyle \frac{x}{x'}})\rangle\neq 0$ from the both sides we find 
\begin{align*}
&\sum_{{\bf i,i',b,b'}\atop a''_1+a'''_1=1,j_1+j'_1=1}\alpha^{a''_1+j_1}\beta^{a'''_1+j'_1}
T({\textstyle \frac{x}{y}})_{\bf i,j}^{\bf a'',b}
T({\textstyle \frac{x'}{y'}})_{\bf i',j'}^{\bf a''',b'}\\
=&\sum_{{\bf k,k',c,c'}\atop a_1+a'_1=1,j''_1+j'''_1=1}\alpha^{a_1+j''_1}\beta^{a'_1+j'''_1}
T({\textstyle \frac{x'}{y'}})_{\bf k',j'''}^{\bf a',c'}
T({\textstyle \frac{x}{y}})_{\bf k,j''}^{\bf a,c}.
\end{align*}
Using \eqref{S} and dividing by $(yy')^2$ we arrive at
\begin{align*}
&\sum_{a''_1+a'''_1=1,j_1+j'_1=1}({\textstyle \frac{x}{y}})^{a''_1+j_1}
({\textstyle \frac{x'}{y'}})^{a'''_1+j'_1}S({\textstyle \frac{x}{y}})_{\bf j}^{\bf a''}
S({\textstyle \frac{x'}{y'}})_{\bf j'}^{\bf a'''}\\
=&\sum_{a_1+a'_1=1,j''_1+j'''_1=1}({\textstyle \frac{x}{y}})^{a_1+j''_1}
({\textstyle \frac{x'}{y'}})^{a'_1+j'''_1}
S({\textstyle \frac{x'}{y'}})_{\bf j'''}^{\bf a'}
S({\textstyle \frac{x}{y}})_{\bf j''}^{\bf a}
\end{align*}
as desired.
\end{proof}

\begin{example}\label{ex:ykw}
When $(m,n)=(1,2)$ and ${\bf j}=(0)$, corollary \ref{cor:S rel} says
\begin{align*}
x^2S(y)^1_{00}S(x)^1_{10} + 
yxS(y)^0_{10}S(x)^1_{00}+
yx S(y)^1_{00}S(x)^0_{10}+
y^2S(y)^1_{10}S(x)^0_{00} = (x \longleftrightarrow y).
\end{align*}
In fact substituting example \ref{ex:nmi} and 
using (\ref{tgm}) we find  that the left hand side is equal to
\begin{align*}
q(x+y)({\bf a}^-\st {\bf k})+
q(x^2+y^2)(1 \st {\bf k})+
qxy(x+y)\bigl(q(1-q){\bf k}\st{\bf a}^+{\bf k}
+ {\bf a}^+\st {\bf k}\bigr)+ xy W
\end{align*}
for some $W \in {\mathscr A}^{\otimes 2}$ independent of $x$ and $y$.
\end{example}

\begin{remark}\label{re:sar}
One can generalize the bilinear relation in theorem \ref{prop:S rel} 
further by introducing {\em inhomogeneity parameters} 
as follows. In \eqref{S} we consider
horizontal lines as carrying parameters $x_1,\dots,x_m$ from the top to the bottom and
vertical lines $y_1,\dots,y_n$ from the left to the right. Set ${\bf x}=(x_1,\dots,x_m),
{\bf y}=(y_1,\dots,y_n)$. We define $S({\bf x};{\bf y})^{\bf a}_{\bf j}$ by changing
the parameter $z$ of $\L(z)$ to $x_i/y_j$ if this $\L(z)$ is situated
on the vertex
where the $i$-th horizontal and the $j$-th vertical line meet.
As in theorem \ref{prop:S rel} let $I,J$ be subsets of
$\{1,\ldots,m\},\{1,\ldots,n\}$
and take $\grbold{\alpha} \in \{0,1\}^{\#I},
\grbold{\gamma} \in \{0,1\}^{\# J}$. 
Suppose that
$({\bf x};{\bf y})$ and $({\bf x'};{\bf y'})$ satisfy
\begin{equation} \label{xy ratio}
x_1/x'_1=\dots=x_m/x'_m=u,\quad y_1/y'_1=\dots=y_n/y'_n=v.
\end{equation}
In other words, $x_i/x'_i$ or $y_i/y'_i$ do not depend on $i$. Then as
a generalization
of theorem \ref{prop:S rel} we have
\begin{equation} \label{mri}
\sum_{\scriptsize\grbold{\beta,\delta}}
({\textstyle\frac{u}{v}})^{|{\scriptsize\grbold{\beta}}|+|{\scriptsize\grbold{\delta}}|}\,
S({\bf x}; {\bf y})^{\scriptsize \grbold{\alpha}_I,\grbold{\beta}_{\overline{I}}}_{
\scriptsize \grbold{\gamma}_J,\grbold{\delta}_{\overline{J}}}\,
S({\bf x'}; {\bf y'})^{\scriptsize \grbold{\alpha}_I,\grbold{\overline{\beta}}_{\overline{I}}}_{
\scriptsize \grbold{\gamma}_J,\grbold{\overline{\delta}}_{\overline{J}}}
=
\sum_{\scriptsize\grbold{\beta,\delta}}
({\textstyle\frac{u}{v}})^{|{\scriptsize\overline{\grbold{\beta}}}|
+|{\scriptsize\overline{\grbold{\delta}}}|}\,
S({\bf x'}; {\bf y'})^{\scriptsize \grbold{\alpha}_I,\grbold{\beta}_{\overline{I}}}_{
\scriptsize \grbold{\gamma}_J,\grbold{\delta}_{\overline{J}}}\,
S({\bf x}; {\bf y})^{\scriptsize \grbold{\alpha}_I,\grbold{\overline{\beta}}_{\overline{I}}}_{
\scriptsize \grbold{\gamma}_J,\grbold{\overline{\delta}}_{\overline{J}}},
\end{equation}
where the sums are over 
$\grbold{\beta} \in \{0,1\}^{\# \overline{I}}$ and 
$\grbold{\delta} \in \{0,1\}^{\# \overline{J}}$ as in (\ref{S rel}).

The proof goes similarly to that of theorem \ref{prop:S rel}. 
We again outline it along 
the claim corresponding to corollary \ref{cor:S rel}. 
By reasoning analogous to
\eqref{MTT=TTM}
we have
\begin{align*}
&\sum_{{\bf b,b'}\atop a''_1+a'''_1=1}
\M({\textstyle \frac{x_1}{x'_1}})_{a''_1,a'''_1}^{a_1,a'_1}
\Bigl(\M({\textstyle \frac{y_n}{y'_n}})_{b_n,b'_n}^{c_n,c'_n}\cdots
\M({\textstyle \frac{y_1}{y'_1}})_{b_1,b'_1}^{c_1,c'_1}\Bigr)
T({\bf x};{\bf y})_{\bf i,j}^{\bf a'',b}T({\bf x'};{\bf y'})_{\bf
i',j'}^{\bf a''',b'}
\\
=&\sum_{{\bf k,k'}\atop j''_1+j'''_1=1}
T({\bf x'};{\bf y'})_{\bf k',j'''}^{\bf a',c'}T({\bf x};{\bf y})_{\bf
k,j''}^{\bf a,c}
\M({\textstyle \frac{y_1}{y'_1}})_{j_1,j'_1}^{j''_1,j'''_1}
\Bigl(\M({\textstyle \frac{x_m}{x'_m}})_{i_m,i'_m}^{k_m,k'_m}\cdots
\M({\textstyle \frac{x_1}{x'_1}})_{i_1,i'_1}^{k_1,k'_1}\Bigr).
\end{align*}
The arguments of $\M$ are determined uniquely so that we can apply theorem
\ref{prop:tetrahedron}. We then wish to multiply a suitable factor and
take sums
over ${\bf i}, {\bf i}', {\bf c}, {\bf c}'$ and 
$a_1,a'_1,j_1,j'_1$ with the constraints $a_1+a'_1=1,j_1+j'_1=1$ on
the both sides of
the above relation. To make the evaluation by $\langle\chi(v^{-1})|$ and
$|\chi(u)\rangle$ successful and 
to get a relation among $S$, we need the condition
\eqref{xy ratio}.
\end{remark}

\section{Application to $n$-TASEP: Proof of theorem \ref{th:mho}}\label{sec:appli}
We set $m=n$ and $q=0$ in the whole construction 
in sections \ref{sec:LM} -- \ref{sec:fb}. 
The resulting objects like ${\bf a}^\pm, {\bf k} \in {\mathscr A}_{q=0}$ and 
$S(z)^{\bf a}_{\bf j} \in({\mathscr A}_{q=0})^{\ot n^2}$
are still well-defined.
The $0$-oscillator generators 
${\bf a}^+,{\bf a}^-,{\bf k} \in {\mathscr A}_{q=0}$  
act on the Fock space as (\ref{lil}) and 
they obey the relations (\ref{rna}).
The 3D $L$ operator $\L(z)$ now defines 
the $0$-oscillator valued five-vertex model:
\begin{equation} \label{5v}
\begin{picture}(200,70)(-10,-45)
\thicklines

\put(-55,-44){$\L(z)^{a,b}_{i,j}=$}

\put(-10,0){\vector(1,0){20}}
\put(0,-10){\vector(0,1){20}}
\put(-17,-3.5){0}\put(12.5,-3.5){0}\put(-2.4,13){0}\put(-2.3,-19){0}
\put(-2,-44){1}

\put(50,0){
\put(-10,0){\color{red} \vector(1,0){20}}
\put(0,-10){\color{red} \vector(0,1){20}}
\put(-17,-3.5){1}\put(12.5,-3.5){1}\put(-2.4,13){1}\put(-2.3,-19){1}
\put(-2,-44){1}}

\put(100,0){
\put(-10,0){\color{red}\line(1,0){10}}\put(0,0){\vector(1,0){10}}
\put(0,-10){\line(0,1){10}}\put(0,0){\color{red}\vector(0,1){10}}
\put(-17,-3.5){1}\put(12.5,-3.5){0}\put(-2.4,13){1}\put(-2.3,-19){0}
\put(-6,-44){$z{\bf a}^+$}}

\put(150,0){
\put(-10,0){\line(1,0){10}}\put(0,0){\color{red}\vector(1,0){10}}
\put(0,-10){\color{red}\line(0,1){10}}\put(0,0){\vector(0,1){10}}
\put(-17,-3.5){0}\put(12.5,-3.5){1}\put(-2.4,13){0}\put(-2.3,-19){1}
\put(-12,-44){$z^{-1}{\bf a}^-$}}

\put(200,0){
\put(-10,0){\vector(1,0){20}}
\put(0,-10){\color{red}\vector(0,1){20}}
\put(-17,-3.5){0}\put(12.5,-3.5){0}\put(-2.4,13){1}\put(-2.3,-19){1}
\put(-2,-44){${\bf k}$}}
\end{picture}
\end{equation}
The other vertices are assigned with 0. 
In particular the rightmost one in (\ref{6v}) vanishes because the 
spectrum of ${\bf k}$ (before setting $q=0$) is given by $(-q)^{\Z_{\ge 0}}$.
See (\ref{kyk}).
At $z=1$ (\ref{5v}) reduces to (\ref{5v0}).

Define the operators $X_0(z), \ldots, X_n(z)$ by 
\begin{equation}\label{Xhn}
\begin{picture}(150,77)(-70,-10)
\put(-110,29){${X}_i(z) =\sum{z^{\alpha_1+\cdots+ \alpha_n}}
$}
\put(20,52){$. . .$}
\put(-5,27){$.$}\put(-5,24){$.$}\put(-5,21){$.$}
\put(-8,48){\line(1,0){56}}
\put(-8,40){\line(1,0){48}}
\put(-8,32){\line(1,0){40}}
\put(-8,16){\line(1,0){24}}
\put(-8,8){\line(1,0){16}}
\put(-8,0){\line(1,0){8}}

\put(11,9.5){\put(29,25){$.$}\put(27,23){$.$}\put(25,21){$.$}}
\put(-9,-9.5){\put(29,25){$.$}\put(27,23){$.$}\put(25,21){$.$}}

\put(48,48){\vector(0,1){8}}\put(44,60){$\scriptstyle{\alpha_n}$}
\put(40,40){\vector(0,1){16}}
\put(32,32){\vector(0,1){24}}
\put(16,16){\vector(0,1){40}}
\put(8,8){\vector(0,1){48}}\put(5,60){$\scriptstyle{\alpha_2}$}
\put(0,0){\vector(0,1){56}}\put(-7,60){$\scriptstyle{\alpha_1}$}

\put(51,46){$\scriptstyle{0}$}
\put(43,38){$\scriptstyle{0}$}
\put(30,25){$\scriptstyle{0}$}

\put(24,18){$\scriptstyle{1}$}
\put(9,2){$\scriptstyle{1}$}
\put(0,-6){$\scriptstyle{1}$}

\put(29,52){\rotatebox{-135}{$\overbrace{\phantom{KKKK}}$}}
\put(54,26){$\scriptstyle{n-i}$}

\put(2,23){\rotatebox{-135}{$\overbrace{\phantom{KKKk}}$}}
\put(26,-2){$\scriptstyle{i}$}

\put(84,29){$\in({\mathscr A}_{q=0})^{\ot n(n-1)/2}$}

\end{picture}
\end{equation}
where all the vertices stand for $\L(z=1)^{a,b}_{i,j}$ in (\ref{5v}) 
(or equivalently (\ref{5v0})) 
and the sum is taken over $\{0,1\}$ for all edges
under the condition that the values along the NE-SW boundary 
are fixed as specified above.
$X_i(z)$ here includes $X_i$ in (\ref{ngm}) as the special case $z=1$.
\begin{example}\label{ex:syr}
We write down $X_0(z), X_1(z), X_2(z)$ for $n=2$ explicitly.
They are $z$-analogue of $X_i$ given in example \ref{ex:mrn2}.
\begin{align*}
X_0(z) = 1+z {\bf a}^+,\quad
X_1(z) = z {\bf k},\quad
X_2(z) = z{\bf a}^-+z^2 1.
\end{align*}
\end{example}

\begin{example}\label{ex:nzm}
We write down $X_0(z), \ldots, X_3(z)$ for $n=3$ explicitly.
They are $z$-analogue of $X_i$ given in example \ref{ex:mrn3}.
\begin{equation*}
\begin{picture}(600,60)(-59,-20)
\setlength\unitlength{0.26mm}
\thicklines

\put(-55,24){$X_0(z)= $}
\put(0,40){\line(1,0){50}} \put(50,40){\vector(0,1){13}}
\put(0,20){\line(1,0){30}} \put(30,20){\vector(0,1){33}}
\put(0,0){\line(1,0){10}}
\put(10,0){\vector(0,1){53}}
\put(-18,-20){$=1\otimes 1 \otimes 1$}

\put(90,0){\put(-20,24){$+$}
\put(10,40){\line(1,0){40}} \put(50,40){\vector(0,1){13}}
\put(0,20){\line(1,0){30}} \put(30,20){\vector(0,1){33}}
\put(0,0){\line(1,0){10}}
\put(10,0){\line(0,1){40}}
\put(0,40){\color{red}\line(1,0){10}}
\put(10,40){\color{red}\vector(0,1){13}}
\put(-22,-20){$+\;\; z{\bf a}^+\!\otimes 1 \otimes 1$}}

\put(185,0){\put(-25,24){$+$}
\put(0,40){\line(1,0){50}} \put(50,40){\vector(0,1){13}}
\put(0,20){\color{red}{\line(1,0){10}}} \put(10,20){\line(1,0){20}}
\put(30,20){\vector(0,1){33}}
\put(0,0){\line(1,0){10}}
\put(10,0){\line(0,1){20}}\put(10,20){\color{red}\vector(0,1){33}}
\put(-22,-20){$+\;\; z{\bf k}\otimes {\bf a}^+\! \otimes 1$}}

\put(290,0){\put(-30,24){$+$}
\put(0,40){\line(1,0){10}} \put(10,40){\color{red}\line(1,0){20}}\put(30,40){\line(1,0){20}} 
\put(50,40){\vector(0,1){13}}\put(10,40){\vector(0,1){13}}\put(10,20){\color{red}\line(0,1){20}}
\put(0,20){\color{red}\line(1,0){10}}
\put(10,20){\line(1,0){20}} \put(30,40){\color{red}\vector(0,1){13}}
\put(30,20){\line(0,1){20}}
\put(0,0){\line(1,0){10}}
\put(10,0){\line(0,1){20}}
\put(-29,-20){$+\;\; z{\bf a}^-\!\otimes {\bf a}^+\! \otimes {\bf a}^+$}}

\put(400,0){\put(-25,24){$+$}
\put(0,20){\color{red}\line(1,0){10}}\put(10,20){\color{red}\vector(0,1){33}}
\put(0,40){\color{red}\line(1,0){30}}\put(30,40){\color{red}\vector(0,1){13}}
\put(30,40){\line(1,0){20}} \put(50,40){\vector(0,1){13}}
\put(10,20){\line(1,0){20}} \put(30,20){\line(0,1){20}}
\put(0,0){\line(1,0){10}}
\put(10,0){\line(0,1){20}}
\put(-25,-20){$+\;\; z^21\otimes {\bf a}^+ \otimes {\bf a}^+$,}}

\end{picture}
\end{equation*}
\begin{equation*}
\begin{picture}(600,60)(-59,-20)
\setlength\unitlength{0.26mm}
\thicklines

\put(-55,24){$X_1(z)= $}
\put(0,40){\line(1,0){50}} \put(50,40){\vector(0,1){13}}
\put(0,20){\line(1,0){30}} \put(30,20){\vector(0,1){33}}
\put(0,0){\color{red}\line(1,0){10}}
\put(10,0){\color{red}\vector(0,1){53}}
\put(-18,-20){$=z{\bf k} \otimes {\bf k} \otimes 1$}

\put(95,0){\put(-20,24){$+$}
\put(10,40){\color{red}\line(1,0){20}} \put(30,40){\line(1,0){20}}
\put(50,40){\vector(0,1){13}}\put(30,40){\color{red}\vector(0,1){13}}
\put(0,20){\line(1,0){30}} \put(30,20){\line(0,1){20}}
\put(0,0){\color{red}\line(1,0){10}}
\put(10,0){\color{red}\line(0,1){40}}
\put(0,40){\line(1,0){10}}
\put(10,40){\vector(0,1){13}}
\put(-22,-20){$+\;\; z{\bf a}^-\!\otimes {\bf k} \otimes {\bf a}^+$}}

\put(200,0){\put(-20,24){$+$}
\put(50,40){\vector(0,1){13}}\put(30,40){\color{red}\vector(0,1){13}}
\put(0,40){\color{red}\line(1,0){30}}\put(30,40){\line(1,0){20}}
\put(30,20){\line(0,1){20}}
\put(0,20){\line(1,0){30}}
\put(10,0){\color{red}\vector(0,1){53}}
\put(0,0){\color{red}\line(1,0){10}}
\put(-22,-20){$+\;\; z^21\otimes {\bf k} \otimes {\bf a}^+$,}}

\end{picture}
\end{equation*}
\begin{equation*}
\begin{picture}(600,60)(-59,-20)
\setlength\unitlength{0.26mm}
\thicklines

\put(-55,24){$X_2(z)= $}
\put(0,40){\line(1,0){50}}\put(50,40){\vector(0,1){13}}
\put(0,20){\line(1,0){10}}\put(10,20){\color{red}\line(1,0){20}}
\put(30,20){\color{red}\vector(0,1){33}}\put(10,20){\vector(0,1){33}}
\put(10,0){\color{red}\line(0,1){20}}
\put(0,0){\color{red}\line(1,0){10}}
\put(-22,-20){$\;=\;z1\otimes {\bf a}^-\! \otimes {\bf k}$}

\put(105,0){\put(-28,24){$+$}
\put(0,40){\color{red}\line(1,0){10}}\put(10,40){\line(1,0){40}}
\put(50,40){\vector(0,1){13}}
\put(0,20){\line(1,0){10}}\put(10,20){\color{red}\line(1,0){20}}
\put(30,20){\color{red}\vector(0,1){33}}\put(10,40){\color{red}\vector(0,1){13}}
\put(10,20){\line(0,1){20}}
\put(10,0){\color{red}\line(0,1){20}}
\put(0,0){\color{red}\line(1,0){10}}
\put(-32,-20){$\;+\;\;z^2{\bf a}^+\!\otimes{\bf a}^-\! \otimes {\bf k}$}
}

\put(205,0){\put(-25,24){$+$}
\put(0,40){\line(1,0){50}} \put(50,40){\vector(0,1){13}}
\put(0,20){\color{red}\line(1,0){30}} \put(30,20){\color{red}\vector(0,1){33}}
\put(0,0){\color{red}\line(1,0){10}}
\put(10,0){\color{red}\vector(0,1){53}}
\put(-21,-20){$+\;\;z^2{\bf k}\otimes 1 \otimes {\bf k}$,}}

\end{picture}
\end{equation*}
\begin{equation*}
\begin{picture}(600,60)(-59,-20)
\setlength\unitlength{0.26mm}
\thicklines

\put(-55,24){$X_3(z)= $}
\put(0,20){\line(1,0){10}}\put(10,20){\vector(0,1){33}}
\put(0,40){\line(1,0){30}}\put(30,40){\vector(0,1){13}}

\put(30,40){\color{red}\line(1,0){20}} \put(50,40){\color{red}\vector(0,1){13}}
\put(10,20){\color{red}\line(1,0){20}} \put(30,20){\color{red}\line(0,1){20}}
\put(0,0){\color{red}\line(1,0){10}}
\put(10,0){\color{red}\line(0,1){20}}
\put(-19,-20){$=z1\otimes {\bf a}^-\! \otimes {\bf a}^-$}

\put(105,0){\put(-30,24){$+$}
\put(0,40){\color{red}\line(1,0){10}} \put(10,40){\line(1,0){20}}
\put(30,40){\color{red}\line(1,0){20}} 
\put(50,40){\color{red}\vector(0,1){13}}\put(10,40){\color{red}\vector(0,1){13}}
\put(10,20){\line(0,1){20}}
\put(0,20){\line(1,0){10}}
\put(10,20){\color{red}\line(1,0){20}} \put(30,40){\vector(0,1){13}}
\put(30,20){\color{red}\line(0,1){20}}
\put(0,0){\color{red}\line(1,0){10}}
\put(10,0){\color{red}\line(0,1){20}}
\put(-29,-20){$+\;\; z^2{\bf a}^+\!\otimes {\bf a}^-\! \otimes {\bf a}^-$}}

\put(218,0){
\put(-30,24){$+$}
\put(0,40){\line(1,0){30}} \put(50,40){\color{red}\vector(0,1){13}}
\put(0,20){\color{red}\line(1,0){30}} \put(30,40){\vector(0,1){13}}
\put(0,0){\color{red}\line(1,0){10}}\put(30,20){\color{red}\line(0,1){20}}
\put(10,0){\color{red}\vector(0,1){53}}\put(30,40){\color{red}\line(1,0){20}}
\put(-28,-20){$+\;z^2{\bf k}\otimes 1 \otimes {\bf a}^-$}}

\put(307,0){\put(-22,24){$+$}
\put(10,40){\color{red}\line(1,0){40}} \put(50,40){\color{red}\vector(0,1){13}}
\put(0,20){\color{red}\line(1,0){30}} \put(30,20){\color{red}\vector(0,1){33}}
\put(0,0){\color{red}\line(1,0){10}}
\put(10,0){\color{red}\line(0,1){40}}
\put(0,40){\line(1,0){10}}
\put(10,40){\vector(0,1){13}}
\put(-22,-20){$+\;\; z^2{\bf a}^-\!\otimes 1 \otimes 1$}}

\put(402,0){\put(-25,24){$+$}
\put(0,40){\color{red}\line(1,0){50}} \put(50,40){\color{red}\vector(0,1){13}}
\put(0,20){\color{red}\line(1,0){30}} \put(30,20){\color{red}\vector(0,1){33}}
\put(0,0){\color{red}\line(1,0){10}}
\put(10,0){\color{red}\vector(0,1){53}}
\put(-21,-20){$+\;z^31\otimes 1\otimes 1$.}}

\end{picture}
\end{equation*}
\end{example}

\begin{proposition} 
The operators $X_i(z)$'s  are contained in 
the layer to layer transfer matrices at $q=0$ as follows:
\begin{equation} \label{S00}
S(z)^{00{\cdots}0}_{00{\cdots}0}=\sum_{i=0}^nX_i(z){\ot}\underbrace {\overbrace {{\bf a}^+{\ot}{\cdots}{\ot}{\bf a^+}}^{i}{\ot}\overbrace{1{\ot}{\cdots}{\ot}1}^{n-i}}_{\rm diagonal}{\ot}1{\ot}{\cdots}{\ot}1,
\end{equation}
\begin{equation} \label{S10}
S(z)^{10{\cdots}0}_{10{\cdots}0}=z^{-1}\sum_{i=0}^nX_i(z){\ot}\underbrace{\overbrace {1{\ot}{\cdots}{\ot}1}^{i}{\ot}\overbrace{{\bf a}^-{\ot}{\cdots}{\ot}{\bf a}^-}^{n-i}}_{\rm diagonal}{\ot}\overbrace{{\bf a}^+{\ot}{\cdots}{\ot}{\bf a}^+}^{n-1}{\ot}1{\ot}{\cdots}{\ot}1.
\end{equation}
Here `diagonal' signifies the part of the tensor components
corresponding to the vertices on the NE-SW diagonal in \eqref{S} with $m=n$.
\end{proposition}
\begin{proof}
We regard the triangle shape region in (\ref{Xhn}) as 
embedded into the 
$n\times n$ square lattice in $(\ref{S})|_{m=n}$.
When $q=0$, the rightmost vertex of $\L(z)$ in (\ref{6v}) is absent. 
This means that the red lines for the allowed configurations 
tend to be confined in the upper left region. 
Also, once an edge on the SW-NE boundary in (\ref{Xhn}) becomes black,
then the subsequent ones continue to be black in its further NE.
These properties imply the claimed expansion formulas.
See the following example from $n=3$, where 
black and red edges are fixed to $0$ and $1$ respectively, 
whereas the dotted ones are to be summed over $0$ and $1$\footnote{
Some of them are actually fixed to $0$ or $1$ by (\ref{5v}),  but 
they are left dotted for the sake of exposition.}.
The four diagrams correspond to $i=0,\ldots, 3$ terms in (\ref{S00}) and 
(\ref{S10}) from the left to the right.  
General case is similar.
\begin{equation*}
\begin{picture}(600,58)(-59,-18)
\setlength\unitlength{0.26mm}
\thicklines

\put(-71,15){$S(z)^{000}_{000}  \ = \ \sum $}
\put(20,-5){
\put(50,40){\vector(0,1){15}}
\put(30,40){\line(1,0){20}}
\put(10,20){\line(1,0){20}}
\put(30,20){\line(0,1){20}}
\put(0,0){\line(1,0){10}}
\put(10,0){\line(0,1){20}}

\put(0,17){$\cdot$}\put(3,17){$\cdot$}
\put(6,17){$\cdot$}
\put(7.8,20){$\cdot$}\put(7.8,23.5){$\cdot$}\put(7.8,26.8){$\cdot$}\put(7.8,30.2){$\cdot$}
\put(7.8,33.3){$\cdot$}\put(7.8,36.5){$\cdot$}\put(7.8,39.5){$\cdot$}
\put(7.8,42){$\cdot$}
\put(7.8,44.5){$\cdot$}
\put(9.7,52){\vector(0,1){3}}

\put(0,36.5){$\cdot$}\put(2.5,36.5){$\cdot$}\put(5,36.5){$\cdot$}
\put(10.5,36.5){$\cdot$}\put(13,36.5){$\cdot$}\put(16.2,36.5){$\cdot$}\put(19.3,36.5){$\cdot$}\put(22.5,36.5){$\cdot$}\put(25.5,36.5){$\cdot$}
\put(28,39.5){$\cdot$}
\put(28,42){$\cdot$}

\put(30,52){\vector(0,1){3}}

\put(10,-10){\line(0,1){10}}
\put(10,0){\line(1,0){20}}
\put(30,-10){\line(0,1){10}}
\put(30,0){\line(0,1){20}}
\put(50,20){\line(0,1){20}}
\put(30,20){\line(1,0){20}}
\put(50,40){\vector(1,0){15}}
\put(50,20){\vector(1,0){15}}
\put(30,0){\vector(1,0){35}}
\put(50,-10){\line(0,1){33}}
}

\put(110,-5){\put(-20,20){$+ \ \sum$}
\put(20,0){
\put(0,17){$\cdot$}\put(3,17){$\cdot$}
\put(6,17){$\cdot$}
\put(7.8,20){$\cdot$}\put(7.8,23.5){$\cdot$}\put(7.8,26.8){$\cdot$}\put(7.8,30.2){$\cdot$}\put(7.8,33.3){$\cdot$}\put(7.8,36.5){$\cdot$}\put(7.8,39.5){$\cdot$}
\put(7.8,42){$\cdot$}

\put(10,52){\vector(0,1){3}}

\put(0,36.5){$\cdot$}\put(2.5,36.5){$\cdot$}\put(5,36.5){$\cdot$}
\put(10.5,36.5){$\cdot$}\put(13,36.5){$\cdot$}\put(16.2,36.5){$\cdot$}\put(19.3,36.5){$\cdot$}\put(22.5,36.5){$\cdot$}\put(25.5,36.5){$\cdot$}
\put(28,39.5){$\cdot$}
\put(28,42){$\cdot$}

\put(30,52){\vector(0,1){3}}

\put(10,-10){\line(0,1){10}}
\put(10,0){\line(1,0){20}}
\put(30,-10){\line(0,1){10}}
\put(30,0){\line(0,1){20}}
\put(50,20){\line(0,1){20}}
\put(30,20){\line(1,0){20}}
\put(50,40){\vector(1,0){15}}
\put(50,20){\vector(1,0){15}}
\put(30,0){\vector(1,0){35}}
\put(50,-10){\line(0,1){30}}

\put(30,40){\line(1,0){20}} \put(50,40){\vector(0,1){15}}
\put(10,20){\line(1,0){20}} \put(30,20){\line(0,1){20}}
\put(0,0){\color{red}\line(1,0){10}}
\put(10,0){\color{red}\line(0,1){20}}
}
}

\put(220,-5){\put(-20,20){$+ \ \sum$}
\put(20,0){
\put(0,17){$\cdot$}\put(3,17){$\cdot$}
\put(6,17){$\cdot$}
\put(7.8,20){$\cdot$}\put(7.8,23.5){$\cdot$}\put(7.8,26.8){$\cdot$}\put(7.8,30.2){$\cdot$}\put(7.8,33.3){$\cdot$}\put(7.8,36.5){$\cdot$}\put(7.8,39.5){$\cdot$}
\put(7.8,42){$\cdot$}

\put(10,52){\vector(0,1){3}}

\put(0,36.5){$\cdot$}\put(2.5,36.5){$\cdot$}\put(5,36.5){$\cdot$}
\put(10.5,36.5){$\cdot$}\put(13,36.5){$\cdot$}\put(16.2,36.5){$\cdot$}\put(19.3,36.5){$\cdot$}\put(22.5,36.5){$\cdot$}\put(25.5,36.5){$\cdot$}
\put(28,39.5){$\cdot$}
\put(28,42){$\cdot$}

\put(30,52){\vector(0,1){3}}

\put(10,-10){\line(0,1){10}}
\put(10,0){\line(1,0){20}}
\put(30,-10){\line(0,1){10}}
\put(30,0){\line(0,1){20}}
\put(50,20){\line(0,1){20}}
\put(30,20){\line(1,0){20}}
\put(50,40){\vector(1,0){15}}
\put(50,20){\vector(1,0){15}}
\put(30,0){\vector(1,0){35}}
\put(50,-10){\line(0,1){30}}

\put(30,40){\line(1,0){20}} \put(50,40){\vector(0,1){15}}
\put(10,20){\color{red}\line(1,0){20}} \put(30,20){\color{red}\line(0,1){20}}
\put(0,0){\color{red}\line(1,0){10}}
\put(10,0){\color{red}\line(0,1){20}}
}
}
\put(330,-5){\put(-20,20){$+ \ \sum$}
\put(20,0){
\put(0,17){$\cdot$}\put(3,17){$\cdot$}
\put(6,17){$\cdot$}
\put(7.8,20){$\cdot$}\put(7.8,23.5){$\cdot$}\put(7.8,26.8){$\cdot$}\put(7.8,30.2){$\cdot$}\put(7.8,33.3){$\cdot$}\put(7.8,36.5){$\cdot$}\put(7.8,39.5){$\cdot$}
\put(7.8,42){$\cdot$}

\put(10,50){\vector(0,1){5}}
\put(0,36.5){$\cdot$}\put(2.5,36.5){$\cdot$}\put(5,36.5){$\cdot$}
\put(10.5,36.5){$\cdot$}\put(13,36.5){$\cdot$}\put(16.2,36.5){$\cdot$}\put(19.3,36.5){$\cdot$}\put(22.5,36.5){$\cdot$}\put(25.5,36.5){$\cdot$}
\put(28,39.5){$\cdot$}
\put(28,42){$\cdot$}

\put(30,52){\vector(0,1){3}}

\put(10,-10){\line(0,1){10}}
\put(10,0){\line(1,0){20}}
\put(30,-10){\line(0,1){10}}
\put(30,0){\line(0,1){20}}
\put(50,20){\line(0,1){20}}
\put(30,20){\line(1,0){20}}
\put(50,40){\vector(1,0){15}}
\put(50,20){\vector(1,0){15}}
\put(30,0){\vector(1,0){35}}
\put(50,-10){\line(0,1){30}}

\put(30,40){\color{red}\line(1,0){20}} \put(50,40){\color{red}\vector(0,1){15}}
\put(10,20){\color{red}\line(1,0){20}} \put(30,20){\color{red}\line(0,1){20}}
\put(0,0){\color{red}\line(1,0){10}}
\put(10,0){\color{red}\line(0,1){20}}
}
}
\end{picture}
\end{equation*}
\begin{equation*}
\begin{picture}(600,55)(-59,-20)
\setlength\unitlength{0.26mm}
\thicklines

\put(-71,15){$S(z)^{100}_{100} \ = \ \sum $}
\put(20,-5){
\put(50,40){\vector(0,1){15}}
\put(30,40){\line(1,0){20}}
\put(10,20){\line(1,0){20}}
\put(30,20){\line(0,1){20}}
\put(0,0){\line(1,0){10}}
\put(10,0){\line(0,1){20}}

\put(0,17){$\cdot$}\put(3,17){$\cdot$}
\put(6,17){$\cdot$}
\put(7.8,20){$\cdot$}\put(7.8,23.5){$\cdot$}\put(7.8,26.8){$\cdot$}\put(7.8,30.2){$\cdot$}\put(7.8,33.3){$\cdot$}\put(7.8,36.5){$\cdot$}\put(7.8,39.5){$\cdot$}
\put(7.8,42){$\cdot$}

\put(10,52){\vector(0,1){3}}

\put(0,36.5){$\cdot$}\put(2.5,36.5){$\cdot$}\put(5,36.5){$\cdot$}
\put(10.5,36.5){$\cdot$}\put(13,36.5){$\cdot$}\put(16.2,36.5){$\cdot$}\put(19.3,36.5){$\cdot$}\put(22.5,36.5){$\cdot$}\put(25.5,36.5){$\cdot$}
\put(28,39.5){$\cdot$}
\put(28,42){$\cdot$}

\put(30,52){\vector(0,1){3}}
\put(10,-10){\color{red}\line(0,1){10}}
\put(10,0){\color{red}\line(1,0){20}}
\put(30,-10){\line(0,1){10}}
\put(30,0){\color{red}\line(0,1){20}}
\put(50,20){\color{red}\line(0,1){20}}
\put(30,20){\color{red}\line(1,0){20}}
\put(50,40){\color{red}\vector(1,0){15}}
\put(50,20){\vector(1,0){15}}
\put(30,0){\vector(1,0){35}}
\put(50,-10){\line(0,1){30}}
}

\put(110,-5){\put(-20,20){$+ \ \sum$}
\put(20,0){
\put(0,17){$\cdot$}\put(3,17){$\cdot$}
\put(6,17){$\cdot$}
\put(7.8,20){$\cdot$}\put(7.8,23.5){$\cdot$}\put(7.8,26.8){$\cdot$}\put(7.8,30.2){$\cdot$}\put(7.8,33.3){$\cdot$}\put(7.8,36.5){$\cdot$}\put(7.8,39.5){$\cdot$}
\put(7.8,42){$\cdot$}

\put(10,52){\vector(0,1){3}}
\put(0,36.5){$\cdot$}\put(2.5,36.5){$\cdot$}\put(5,36.5){$\cdot$}
\put(10.5,36.5){$\cdot$}\put(13,36.5){$\cdot$}\put(16.2,36.5){$\cdot$}\put(19.3,36.5){$\cdot$}\put(22.5,36.5){$\cdot$}\put(25.5,36.5){$\cdot$}
\put(28,39.5){$\cdot$}
\put(28,42){$\cdot$}

\put(30,52){\vector(0,1){3}}
\put(10,-10){\color{red}\line(0,1){10}}
\put(10,0){\color{red}\line(1,0){20}}
\put(30,-10){\line(0,1){10}}
\put(30,0){\color{red}\line(0,1){20}}
\put(50,20){\color{red}\line(0,1){20}}
\put(30,20){\color{red}\line(1,0){20}}
\put(50,40){\color{red}\vector(1,0){15}}
\put(50,20){\vector(1,0){15}}
\put(30,0){\vector(1,0){35}}
\put(50,-10){\line(0,1){30}}

\put(30,40){\line(1,0){20}} \put(50,40){\vector(0,1){15}}
\put(10,20){\line(1,0){20}} \put(30,20){\line(0,1){20}}
\put(0,0){\color{red}\line(1,0){10}}
\put(10,0){\color{red}\line(0,1){20}}
}
}
\put(220,-5){\put(-20,20){$+ \ \sum$}
\put(20,0){
\put(0,17){$\cdot$}\put(3,17){$\cdot$}
\put(6,17){$\cdot$}
\put(7.8,20){$\cdot$}\put(7.8,23.5){$\cdot$}\put(7.8,26.8){$\cdot$}\put(7.8,30.2){$\cdot$}\put(7.8,33.3){$\cdot$}\put(7.8,36.5){$\cdot$}\put(7.8,39.5){$\cdot$}
\put(7.8,42){$\cdot$}

\put(10,52){\vector(0,1){3}}
\put(0,36.5){$\cdot$}\put(2.5,36.5){$\cdot$}\put(5,36.5){$\cdot$}
\put(10.5,36.5){$\cdot$}\put(13,36.5){$\cdot$}\put(16.2,36.5){$\cdot$}\put(19.3,36.5){$\cdot$}\put(22.5,36.5){$\cdot$}\put(25.5,36.5){$\cdot$}
\put(28,39.5){$\cdot$}
\put(28,42){$\cdot$}

\put(30,52){\vector(0,1){3}}

\put(10,-10){\color{red}\line(0,1){10}}
\put(10,0){\color{red}\line(1,0){20}}
\put(30,-10){\line(0,1){10}}
\put(30,0){\color{red}\line(0,1){20}}
\put(50,20){\color{red}\line(0,1){20}}
\put(30,20){\color{red}\line(1,0){20}}
\put(50,40){\color{red}\vector(1,0){15}}
\put(50,20){\vector(1,0){15}}
\put(30,0){\vector(1,0){35}}
\put(50,-10){\line(0,1){30}}

\put(30,40){\line(1,0){20}} \put(50,40){\vector(0,1){15}}
\put(10,20){\color{red}\line(1,0){20}} \put(30,20){\color{red}\line(0,1){20}}
\put(0,0){\color{red}\line(1,0){10}}
\put(10,0){\color{red}\line(0,1){20}}
}
}
\put(330,-5){\put(-20,20){$+ \ \sum$}
\put(20,0){
\put(0,17){$\cdot$}\put(3,17){$\cdot$}
\put(6,17){$\cdot$}
\put(7.8,20){$\cdot$}\put(7.8,23.5){$\cdot$}\put(7.8,26.8){$\cdot$}\put(7.8,30.2){$\cdot$}\put(7.8,33.3){$\cdot$}\put(7.8,36.5){$\cdot$}\put(7.8,39.5){$\cdot$}
\put(7.8,42){$\cdot$}

\put(10,52){\vector(0,1){3}}

\put(0,36.5){$\cdot$}\put(2.5,36.5){$\cdot$}\put(5,36.5){$\cdot$}
\put(10.5,36.5){$\cdot$}\put(13,36.5){$\cdot$}\put(16.2,36.5){$\cdot$}\put(19.3,36.5){$\cdot$}\put(22.5,36.5){$\cdot$}\put(25.5,36.5){$\cdot$}
\put(28,39.5){$\cdot$}
\put(28,42){$\cdot$}

\put(30,52){\vector(0,1){3}}
\put(10,-10){\color{red}\line(0,1){10}}
\put(10,0){\color{red}\line(1,0){20}}
\put(30,-10){\line(0,1){10}}
\put(30,0){\color{red}\line(0,1){20}}
\put(50,20){\color{red}\line(0,1){20}}
\put(30,20){\color{red}\line(1,0){20}}
\put(50,40){\color{red}\vector(1,0){15}}
\put(50,20){\vector(1,0){15}}
\put(30,0){\vector(1,0){35}}
\put(50,-10){\line(0,1){30}}

\put(30,40){\color{red}\line(1,0){20}} \put(50,40){\color{red}\vector(0,1){15}}
\put(10,20){\color{red}\line(1,0){20}} \put(30,20){\color{red}\line(0,1){20}}
\put(0,0){\color{red}\line(1,0){10}}
\put(10,0){\color{red}\line(0,1){20}}
}
}
\end{picture}
\end{equation*}
For the weight of $z$, notice that it is calculated by
\#($1$ on the top edges) $-$\#($1$ on the bottom edges).
\end{proof}

\begin{example}\label{ex:mkr}
Consider the case $n=2$.
Setting $q=0$ in example \ref{ex:ts}, we have
\begin{align*}
S(z)^{00}_{00} &=
(1+z {\bf a}^+)\st 1 \st 1 \st 1+ 
z{\bf k}\st {\bf a}^+ \st 1 \st 1+
(z{\bf a}^-+z^21)\st {\bf a}^+ \st {\bf a}^+ \st 1\\ 
&= X_0(z)\st 1 \st 1 \st 1 +
X_1(z)\st {\bf a}^+ \st 1 \st 1+
X_2(z) \st {\bf a}^+ \st {\bf a}^+ \st 1
\end{align*}
by example \ref{ex:syr} in agreement with (\ref{S00}).
Similarly example \ref{ex:ts2} leads to
\begin{align*}
zS(z)^{10}_{10} &=
(1+z {\bf a}^+)\st {\bf a}^- \st {\bf a}^- \st {\bf a}^+ +
z{\bf k} \st 1 \st {\bf a}^- \st {\bf a}^+ +
(z{\bf a}^-+z^21)\st 1 \st 1 \st {\bf a}^+\\
&=X_0(z) \st {\bf a}^- \st {\bf a}^- \st {\bf a}^+ +
X_1(z)\st 1 \st {\bf a}^- \st {\bf a}^+ +
X_2(z)\st 1 \st 1 \st {\bf a}^+
\end{align*}
in agreement with (\ref{S10}).
\end{example}

Now we are going to extract the relations among $X_i(z)$'s 
from the $q=0$ limit of the bilinear identities in 
proposition \ref{prop:S comm} and corollary \ref{cor:S rel}.

\begin{proposition}[Difference analogue of the hat relation]\label{pr:akn}
The operators $X_i(z)$'s 
satisfy the following relations:
\begin{align}
[X_i(x),X_j(y)]&=[X_i(y),X_j(x)]\qquad (0 \le i,j \le n),
\label{XX=XX}\\
xX_i(y)X_j(x)&=yX_i(x)X_j(y)\qquad \;(0 \le j<i \le n).
\label{xXX}
\end{align}
\end{proposition}
\begin{proof}
Substituting \eqref{S00} into \eqref{S comm} and taking the coefficient of 
\[
\overbrace{({\bf a}^+)^2{\ot}{\cdots}{\ot}({\bf a}^+)^2}^{j}{\ot}
\overbrace{{\bf a}^+{\ot}{\cdots}{\ot}{\bf a}^+}^{i-j}
{\ot}1{\ot}{\cdots}{\ot}1\qquad (0 \le j \le i \le n),
\] 
we get (\ref{XX=XX}).
Set
${\bf a}=(0,{\ldots},0)$, ${\bf j}=(0,{\ldots},0)$ in 
corollary \ref{cor:S rel} and use the obvious property  
$S(z)^{10\cdots0}_{00{\cdots}0}=S(z)^{00\cdots0}_{10{\cdots}0}=0$ 
to derive
\begin{equation*}
x^2S(y)^{00{\cdots}0}_{00{\cdots}0}S(x)^{10{\cdots}0}_{10{\cdots}0}
+y^2S(y)^{10{\cdots}0}_{10{\cdots}0}
S(x)^{00{\cdots}0}_{00{\cdots}0}=(x \longleftrightarrow y).
\end{equation*}
Substitute \eqref{S00}, \eqref{S10} into this and take the coefficient of  
\[
\overbrace{{\bf a}^+{\ot}{\cdots}{\ot}{\bf a}^+}^{j}
{\ot}\overbrace{{\bf k}{\ot}{\cdots}{\ot}{\bf k}}^{i-j}
{\ot}\overbrace{{\bf a}^-{\ot}{\cdots}{\ot}{\bf a}^-}^{n-i}
\ot(\text{off diagonal})\qquad(0 \le j < i \le n).
\]
Noting that such term comes only from 
$(\overbrace{{\bf a}^+{\ot}{\cdots}
{\ot}{\bf a}^+}^{i}{\ot}1{\ot}{\cdots}{\ot}1)
(\overbrace{1{\ot}{\cdots}{\ot}1}^{j}
{\ot}{\bf a}^-{\ot}{\cdots}{\ot}{\bf a}^-)$, we obtain (\ref{xXX}).
\end{proof}

\begin{remark}\label{re:sae}
The relations in proposition \ref{pr:akn} are rearranged as
\begin{align*}
X_i(x)X_j(y) = \begin{cases}
X_i(y)X_j(x) + (1-\frac{x}{y})X_j(y)X_i(x) & i<j,\\
X_i(y)X_i(x) & i=j,\\
\frac{x}{y}X_i(y)X_j(x) & i>j.
\end{cases}
\end{align*}
This exchange rule satisfies the Yang-Baxter relation in that 
the two ways of rewriting $X_i(x)X_j(y)X_k(z)$ 
as linear combinations of 
$X_{k'}(z)X_{j'}(y)X_{i'}(x)$ with 
$\{i',j',k' \} = \{i,j,k\}$ lead to the identical result.
They are equivalent to 
the $t=0$ case of eqs. (25) and (26) in \cite{CDW}
under the formal correspondence 
$X_i(z) = A_{n-i}(z^{-1})$.
\end{remark}

Finally we introduce the $n$-TASEP operators 
$X_i, {\hat X}_i \in ({\mathscr A}_{q=0})^{\ot n(n-1)/2}$ \cite{KMO} by
\begin{equation}\label{ykn}
X_i=X_i(z=1), \quad  {\hat X}_i=\frac{d}{dz}X_i(z) |_{z=1}
\qquad (0 \le i \le n).
\end{equation}
From (\ref{Xhn}) we see that they coincide with 
those defined in (\ref{ngm}) 
as the configuration sums of the $0$-oscillator valued five-vertex model
whose vertices are specified in (\ref{5v0}).

\vspace{0.3cm}
{\it Proof of theorem \ref{th:mho}}. 
Differentiate \eqref{XX=XX} and  \eqref{xXX}  with
respect to $y$ and set $x,y=1$.
\qed

\section{Summary}\label{sec:sum}
In this paper we have proved the hat relation in theorem \ref{th:mho} among 
the operators $X_i$ and ${\hat X}_i$ defined by (\ref{ngm}).
It provides an alternative derivation of 
the matrix product formula for the steady state probability (\ref{mho})
of the $n$-TASEP, which was obtained earlier in \cite{KMO}
by identifying the Ferrari-Martin algorithm with a composition of the 
combinatorial $R$.

Reversing the order of presentation in this paper, 
our proof of the hat relation may be summarized as follows.
The hat relation (theorem \ref{th:mho}) is first upgraded to 
the difference analogue in proposition \ref{pr:akn}.
By introducing $q$ and embedding into the 3D lattice model,
it is further upgraded to bilinear relations among 
layer to layer transfer matrices (theorem \ref{prop:S rel}).
Finally these relations are attributed to the most local property, 
the tetrahedron equation in proposition \ref{prop:tetrahedron}.

The present paper and \cite{KMO} reveal a hidden 3D integrable structure
in the multispecies TASEP.
It deserves further investigation whether such results can be generalized 
to the large list of matrix product constructions 
of the quantum and combinatorial $R$ by the tetrahedron equation \cite{KOS,Ku}.
It turns out that another prototype model of stochastic dynamics
known as the multispecies 
{\em totally asymmetric zero range process}  (so called TAZRP)
can be analyzed in a completely similar manner based on the scheme given
in this paper.
We plan to present the detail in a future publication.

\section*{Acknowledgments}
A. K. thanks organizers of ``Baxter 2015: Exactly Solved Models $\&$ Beyond",
Palm Cove, Australia, 19-25 July 2015 for warm hospitality, 
where a part of the work was presented.
This work is supported by 
Grants-in-Aid for Scientific Research No.~15K04892,
No.~15K13429 and No.~23340007 from JSPS.

\end{document}